\documentclass[12pt]{article}

\usepackage[a4paper,top=2cm,left=2cm,right=1.5cm,bottom=2.5cm]{geometry}
\usepackage{beamerarticle}

\usepackage{booktabs}
\usepackage{amsmath}
\usepackage{amssymb}
\usepackage{multirow}
\usepackage{multicol}
\usepackage{mathtools}
\usepackage{bm}
\usepackage{float}
\usepackage{adjustbox}
\usepackage[T1]{fontenc}

\usepackage[numbered,framed]{matlab-prettifier}
\lstMakeShortInline"
\newcommand{\E}{\mathbb{E}}

\newcommand{\plim}{\text{plim}}

\newtheorem{assumption}{Assumption}

\usepackage{natbib}
\usepackage[colorlinks=true, linkcolor=blue, citecolor=blue, urlcolor=blue]{hyperref}

\usepackage{setspace}
\begin{document}


\begin{titlepage}
    \title{Targeted Local Projections}
    \author{Aleksei Nemtyrev\thanks{Department of Econometrics \& OR, Tilburg University.} \and Otilia Boldea\footnotemark[1] \, \thanks{Corresponding author, o.boldea@tilburguniversity.edu}}
    \date{\today}
    \maketitle
       \vspace{-0.2in}
    \begin{abstract}
 
    \noindent 
    Local projections (LPs) and structural vector autoregressions (SVARs) are commonly employed to estimate dynamic causal effects of macroeconomic policies. With enough lags as controls, LP estimators have little bias but their variance can increase with the horizon. Because of employing fewer lags or due to local misspecification, SVAR estimators typically incur higher bias, but their variance decreases with the horizon due to exponentiation. We propose to target the LP estimators towards their SVAR counterparts - constructed with fewer lags than LP at each horizon - to reduce their variance at the cost of incurring some bias. The resulting estimator is a linear combination of the LP and SVAR estimators, and we propose choosing it optimally to minimize the mean-squared error of the new estimator. Our simulations show that, under a locally misspecified SVAR model, targeting substantially reduces the LP variance at longer horizons while reducing coverage distortions relative to SVAR when a double bootstrap is employed. In an empirical application, we estimate the impact of a monetary policy shock on macroeconomic aggregates. For many of these aggregates, the corresponding TLP estimates are more precise than their LP counterparts.    
    
     \vspace{0.2in}
    
    \vspace{0.1in}
    
    
    \bigskip
    \end{abstract}
    \setcounter{page}{0}
    \thispagestyle{empty}
\end{titlepage}
    \pagebreak \newpage

\section{Introduction}
Identifying and estimating dynamic causal effects of structural shocks on macroeconomic outcomes is a fundamental goal in macroeconomics. To study these dynamic causal effects, researchers use a variety of tools, including taking to the data theoretical models such as dynamic stochastic general equilibrium models, and other models such as structural vector autoregressions (SVARs) and local projections (LPs).

Local projections were popularized by \cite{jorda2005estimation}\footnote{An earlier paper by \cite{dufour1998short} features a similar idea.} and are widely used to estimate impulse response functions (IRFs) - which coincide with the dynamic causal effects of a structural shock in linear models under appropriate conditioning on lags. A recent overview of the use and inference with LPs is provided in \cite{jorda2023local}, \cite{jorda2025local}, \cite{olea2025local} and \cite{inoue2026inference}. LPs are more robust to misspecification than SVARs, as shown in \cite{olea2024double}, and are semi-parametrically efficient under classical assumptions if the number of lags included diverges, as shown in \cite{xu2023local}. Therefore, if one is interested exclusively in accurate coverage, either LPs or SVARs with a large number of lags should be employed, as also argued in \cite{olea2025local}.

In this paper, we are interested in pointwise estimation instead. In that case, it is possible to reduce the variance of the LP estimates at the cost of a small bias. This may be useful for applied work because the original LP estimators exhibit variances that increase with the horizon, as suggested by the results in \cite{lusompa2023local}, and this is due to the LP regression rather than the true data generating process. The excessive variability of LPs across horizons has prompted researchers to smooth impulse responses across horizons. \cite{plagborg2016essays} proposed smoothing the adjacent changes in LP estimators towards each other, after the LP estimators are obtained, and therefore in two steps. \cite{barnichon2019impulse} also proposed smooth local projection (SLP) estimators that penalize changes across nearby horizons, but the smoothing is done in one step, via a penalized least-squares objective function. However, smoothing nearby horizons does not guarantee an exponential decay similar to that one would expect if the true data generating process was a stationary SVAR of finite or infinite order. 

In contrast, \cite{ferreira2023bayesian} propose achieving exponential decay by shrinking the LP estimates towards a VAR in a Bayesian framework. To obtain their Bayesian LP (BLP) estimator, they employ a quasi-Bayesian LP approach, where the prior mean of the coefficient is equal to a VAR estimator on a presample, while the prior variance is based on a Minnesota prior, with a shrinkage parameter that maximizes the posterior marginal density of the data. However, due to small samples in macroeconomics, presample VAR estimates are rarely reliable. Additionally, \cite{ferreira2023bayesian} assume that the true data generating process is a finite order VAR, while in reality, it is likely that the model is misspecified - \cite{olea2024double}. To address this issue, \cite{gonzalez2025misspecification} propose a similar Bayesian procedure, but the shrinkage parameter minimizes the prediction error and their prediction error criterion can also be used to choose among LP and VAR estimators. Their focus is on reduced-form estimation and on multi-step forecasts rather than structural parameters. For point estimation of target structural parameters, it is unclear that minimizing prediction error would be desirable.

Compared to smooth local projections \citep{plagborg2016essays, barnichon2019impulse}, TLP differs in both its regularization target and its treatment of horizons. SLP penalizes differences across adjacent horizons, shrinking the IRFs toward each other, without a data-driven target, while TLP shrinks LPs toward the SVAR IRFs. SLP also applies a single smoothing parameter across all horizons and creates distortions at boundary horizons; TLP allows shrinkage to vary by horizon and does not directly impose cross-horizon smoothness. Compared to BLP \citep{ferreira2023bayesian}, TLP uses the full sample for both LP and SVAR estimation, whereas BLP relies on a presample VAR as the prior. Additionally, BLP assumes correct VAR specification and our simulations show that under misspecification, it exhibits severe coverage distortions at shorter horizons. Finally, TLP is a linear combination of LP and SVAR with transparent, data-driven weights, while BLP is not a linear combination of the two.

There are a few other papers that also consider taking linear combinations of LP and SVAR. \cite{ho2024} consider prediction pools where multiple estimators are averaged to minimize the weighted average log score function for forecasts, conditional on the structural shock of interest. \cite{hounyo2025} propose to first take horizon-specific linear combinations of multiple LP and SVAR estimators separately, then take a linear combination of these two, with weights that minimize prediction error or $R^2$. The motivation for TLP is different from that of \cite{ho2024} and \cite{hounyo2025}: we want to reduce the variance of the LP estimators, at the cost of incurring a small bias, and therefore the risk function we minimize is the mean-squared error of TLP. In another related but distinct contribution, \cite{chen2026estimator} also propose a linear combination of LP and SVAR estimators to minimize the risk of the resulting estimator - variance or mean-squared error; however, their asymptotic theory is in the context in which both LP and SVAR do not exhibit asymptotic bias. They start from a general data generating process that can be approximated with sufficiently many lags that diverge with the sample size at the appropriate rate, and employ a sieve AR bootstrap with bias correction to estimate their proposed horizon-specific weights. Compared to our estimator,  \cite{chen2026estimator} focus on finite sample rather than asymptotic trade-offs; therefore, their approach complements ours.

In the case of local misspecification, the SVAR estimates will exhibit asymptotic bias, distorting coverage of TLP as well. As discussed in \cite{olea2025local}, impossibility results from microeconometrics - \cite{pratt1961length}, \cite{armstrong2018optimal} - suggest that coverage distortions may be unavoidable.\footnote{\cite{armstrong2018optimal} provide a general procedure for adaptive estimators to balance efficiency and robustness, and suggest reporting both best- and worst-case coverage, but do not consider time series set-ups.}  However, the coverage distortions of TLP will be smaller than those of the SVAR because the bias of TLP is smaller. 

For inference, we propose a Mean Symmetric Double Bootstrap (MSDB) procedure that 
provides near-nominal coverage for LP, SLP, and TLP when the SVAR misspecification is sufficiently small, as shown in our simulations. The procedure uses a nested bootstrap structure: the first level 
generates bootstrap samples via moving-block resampling of SVAR residuals, while the 
second level provides variance estimates needed for studentization, which are necessary 
because the delta method performs poorly for SVAR and TLP impulse responses in finite samples. Two innovations are central to minimizing coverage distortions for TLP and VAR. First, we center the 
bootstrap t-statistics around the bootstrap mean rather than the SVAR-implied impulse 
response to remove further finite sample distortions. Second, we construct symmetric confidence intervals 
\citep{hall1988symmetric}, which help to correct finite sample asymmetry in the VAR and TLP distributions. 

It is worth noting that this bootstrap replicates the variance of the SVAR and TLP estimates, but not the asymptotic bias when the SVAR is locally misspecified. Intuitively, the first level bootstrap is centered at the SVAR estimates, so the difference between the bootstrap estimates and the SVAR estimates is centered around zero, while the difference between SVAR and TLP estimates and the true IRFs are centered at their respective asymptotic bias. Therefore, the MSDB studentized confidence intervals will still be shifted due to bias; however, the shift is smaller for TLP because its asymptotic bias is smaller. In simulations featuring 
recursively identified SVARs with local misspecification, we demonstrate that 
TLP substantially reduces variance relative to LP and mean-squared error at longer horizons relative to both LP and SVAR, while maintaining superior coverage to that of the SVAR, and near-nominal coverage provided the misspecification is not too large. 

We revisit the empirical application in \cite{ferreira2023bayesian} and estimate the IRFs of a monetary shock on real GDP, real consumption, real investment, total hours worked, real wages and the GDP
deflator. For several but not all of these variables, the LP and SVAR estimates are markedly different at higher horizons, suggesting that the SVAR may be misspecified, as also noted in \cite{ferreira2023bayesian}. In those cases, TLP estimates still have smaller variance than LP, even though more weight is placed on LP. When the LP and SVAR estimates are close to each other, such as for IRFs of a monetary policy shock on hours worked, more weight is placed on SVAR and the TLP variance reductions relative to LP are larger.

The remainder of this paper proceeds as follows. Section~\ref{sec:model} presents 
the baseline estimators and introduces the targeted local projections estimator. 
Section~\ref{sec:select_lambda} derives the optimal shrinkage parameter by 
minimizing mean-squared error. Section~\ref{sec:variance} describes the bootstrap procedure. Section~\ref{sec:simulation} evaluates the finite-sample performance 
of TLP relative to existing methods. Section~\ref{sec:application} contains the empirical application. Section~\ref{sec:conclusion} concludes. Appendix~\ref{app:controls} discusses the importance of projecting out controls before targeting. Appendix \ref{app:proofs} contains proofs, Appendix \ref{app:sims} further simulations, and Appendix \ref{app:empirics} further empirical results.

\section{Model set-up\label{sec:model}}
\subsection{Data Generating Process}
We consider the data generating process to be a local-to-VAR model\footnote{We refer to the SVAR as VAR from now onward.} that is locally misspecified as in \cite{olea2024double}:
\begin{align}
    y_{t} = A y_{t-1} + \Gamma\left(I+ T^{-\zeta} \alpha(L) \right) \epsilon_t, \label{DGP}
\end{align}
where $y_t = (y_{1,t},\dots,y_{k,t})'$ is the data, the structural shocks are $\epsilon_t=(\epsilon_{1,t},\dots,\epsilon_{k,t})'$, \(\alpha(L) = \sum_{\ell=1}^{\infty} \alpha_\ell L^\ell\) is a \(k \times k\) lag polynomial, and \(T\) denotes the sample size. The parameter of interest is the impulse response at horizon $h$ of the variable $y_{i,t+h}$ to the shock $\epsilon_{j,t}$ for some indices $i,j \in \{1,\dots,k\}$. The assumption below is adapted from \cite{olea2024double}, setting for simplicity the number of shocks equal to the number of variables.\footnote{This can be generalized to having more or fewer shocks as long as the shocks of interest are identified via appropriate ordering.} 

\begin{assumption}\label{asn:1}
For each $T$, $\{y_t\}_{t\in\mathbb{Z}}$ is the stationary solution to equation (\ref{DGP}), given the following restrictions on parameters and shocks:
\begin{enumerate}[i)]
    \item $\epsilon_t$ is a strictly stationary martingale difference sequence with respect to the filtration $\mathcal F_{t-1}=\sigma \{\epsilon_{t-1}, \epsilon_{t-2}, \ldots\}$; $ \mathrm{Var}(\epsilon_t) = D = \mathrm{diag}(\sigma_1^2, \ldots \sigma^2_k)$; $\sigma_j^2>0$ for all $j=1,\ldots, k$; $\mathbb{E} \|\epsilon_t\|^{d}< \infty$ for $d=4+2\delta$, for some $\delta>0$; $\sum_{\ell,\tau=1}^{\infty} \| \mathrm{Cov} (\epsilon_t \otimes \epsilon_t, \epsilon_{t-\ell} \otimes \epsilon_{t-\tau} )\| <\infty$.
    \item All eigenvalues of $A$ are inside the unit circle and different than zero.
    \item $\Gamma$ is a $k \times k$ lower triangular matrix with $1$'s on the diagonal. 
    \item $S = \mathrm{Var}(\tilde{y}_t)$ is non-singular, where $\tilde{y}_t = (I - AL)^{-1} \Gamma \epsilon_t$ is the stationary solution to (\ref{DGP}) when $\alpha(L) = 0$. 
    \item $\alpha(L)$ is absolutely summable.
    \item $\zeta > 1/4$. 
\end{enumerate}
\end{assumption}

This assumption ensures that the process $\tilde{y}_t$ around which the VAR process $y_t$ is locally misspecified is stationary. Assumption \ref{asn:1}(iii) ensures that the dynamic causal effects of the shocks $\epsilon_{i,t}$ are identified through ordering, a common assumption for identification.

The impulse response at horizon $h=\{0,\dots,H\}$ is defined as 
\begin{align}
\beta_h &= 
E\left[ y_{i,t+h} \mid \epsilon_{j,t} = 1 \right]
-
E\left[ y_{i,t+h} \mid \epsilon_{j,t} = 0 \right] \nonumber \\
&= J_i' 
\Bigg(
A^{h} \Gamma
+
T^{-\zeta} \sum_{\ell=1}^{h} A^{\,h-\ell} \Gamma \alpha_{\ell}
\Bigg) 
J_j,\label{BetaStar}
\end{align}
where $J_i$ and $J_j$ are the selection vectors of dimensions $k \times 1$ with 1 in the position of $i$ or $j$ and $0$ otherwise. Note that the IRF we consider at each horizon is a scalar. If the true data generating process (DGP) has more than one lag $q>1$, we can transform it into a local-to-VAR($1$) using the companion matrix, and therefore consider without loss of generality the local-to-VAR(1) data generating process. We also omit in the DGP above the intercept for simplicity; however, this can be included without further complications besides additional notation.

\subsection{Baseline Estimators}
\label{sec:baseline}

\textbf{Local projections (LP).} 
The LP estimator is the coefficient $\hat{\beta}_h^{LP}$ in the regression
\begin{align}
    y_{i,t+h} &= \beta_h y_{j,t}  +  \sum_{s=1}^k \sum_{\tau=1}^p \gamma_{s,t-\tau} y_{s,t-\tau} + u^h_{i,t+h}.
\end{align}
where $u^h_{i,t+h}$ is the LP error, and $p$ is the number of lags included as controls. Stacking $y_{i,t+h}$, 
$y_{j,t}$, controls and $u^h_{i,t+h}$ over time in the vectors 
${Y}_{i,t+h}$, ${Y}_{j,t}$, $W$ and $U^h_{i,t+h}$ yields
\begin{align}
    {Y}_{i,h} = {Y}_{j} \beta_h + W \delta_h + U^h_i,
\end{align}
where $W$ is a $(T-h)\times (kp)$ matrix collecting the lagged controls, and $\delta_h$ stacks $\gamma_{s,t-\tau}$ appropriately over $s,\tau$. Note that the LP estimator we consider here is lag-augmented with $p \geq 1$ lags.

Projecting out controls from both $Y_{i,h}$ and $Y_j$:
\begin{equation}
    \underbrace{M_W Y_{i,h}}_{Y_h} = \underbrace{M_W Y_{j}}_{X} \beta_h 
    + \underbrace{M_W U^h_i}_{e_h},
\end{equation}
where $M_W = I - W(W'W)^{-1}W'$ yields the simplified model:
\begin{equation}
    Y_h = X \beta_h + e_h.
    \label{eq:LP_projected}
\end{equation}
Projecting out controls is not strictly necessary for LP estimation, but we do so because the coefficients $\delta_h$ are typically nuisance parameters, and we only care about the dynamic causal effects. Regularizing the latter without penalizing the nuisance parameters induces further bias in the coefficients on the control variables without any clear benefit.\footnote{Appendix~\ref{app:controls} clarifies the source of this additional bias.}

\vspace{1em}
\textbf{Structural vector autoregression (VAR).}
The SVAR estimator $\hat{\beta}^{VAR}_h$ is the response of $y_{i,t+h}$ to 
the $j$-th innovation, obtained via Cholesky decomposition:
\begin{align}
    \hat{\beta}^{VAR}_h = J_i' \hat{A}^h \hat{\Gamma} J_j,
    \label{eq:BetaVAR}
\end{align}
where $\hat{A}$ is the VAR coefficient matrix estimate obtained via OLS equation by equation.
$
    \hat{A} = \left( \sum_{t=1}^{T} y_t y_{t-1}' \right)
    \left( \sum_{t=1}^{T} y_{t-1} y_{t-1}' \right)^{-1},
$
and $\hat{\Gamma}$ is the lower triangular Cholesky factor of the residual 
covariance matrix $\hat{\Sigma}^{VAR} = \frac{1}{T} \sum_{t=1}^T 
\hat{u}_t \hat{u}_t'$, with $\hat{u}_t = y_t - \hat{A} y_{t-1}$. 

\vspace{1em}
\textbf{Asymptotic properties.}
Under the local-to-VAR framework with $\zeta > 1/4$, \cite{olea2024double} show that LP is doubly 
robust in the sense that it has no asymptotic bias. In contrast, VAR suffers from asymptotic bias when $\zeta \in (1/4,1/2]$. This is stated in their Propositions 3.1 and 3.2, which assume that the number of lags is the same in the LP and VAR and equal to that in the true DGP. However, the result generalizes to the case in which LP contains more lags than the true number of lags, and this is stated in Theorem \ref{thm:1} below.
\begin{theorem} \label{thm:1} Under Assumption \ref{asn:1}, 
$\mathbb E(\hat{\beta}_h^{LP} - \beta_h) = o(T^{-1/2}),$
while

$\mathbb E(\hat{\beta}_h^{VAR} - \beta_h) = T^{-\zeta} \, \text{aBias}_h + o(T^{-1/2}+T^{-\zeta}),$
where
\begin{align*}
    \text{aBias}_h &= \text{tr}\Bigg( S^{-1} \Psi_h \Gamma 
    \sum_{\ell=1}^{\infty} \alpha_{\ell} D \Gamma' (A')^{\ell-1} \Bigg) 
    - J_{i}' \sum_{\ell=1}^{h} A^{h-\ell} \Gamma \alpha_{\ell} J_{j}, \\
    \Psi_h &= \sum_{\ell=1}^{h} A^{h-\ell} \Gamma_{j} J_{i}' A^{\ell-1}.
\end{align*}    
\end{theorem}
When $1/4 < \zeta \leq 1/2$, $\mathrm{E}(\hat \beta_h^{VAR}) - \beta_h = O(T^{-\zeta})$, and the VAR estimator has an asymptotic bias, while when $\zeta > 1/2$, the VAR does not have an asymptotic bias. 

Next, we impose mixing conditions on the structural errors.

\begin{assumption}\label{asn:2}
Assume that $\epsilon_t$ is an $\alpha$-mixing process of size $-r/(r-2)$, for $r=2+\delta$ and the same $\delta$ as in Assumption \ref{asn:1}. 
\end{assumption}

The mixing condition above is used to verify the WLLN and CLT for near-epoch dependent processes (also mixingales) and to show that in this case, the asymptotic LP variance does not depend to first order on the misspecification.\footnote{Note that, given the misspecification, the martingale difference assumption on the errors is insufficient to verify the CLT, and, as \cite{xu2023local} notes, more moment bounds are needed to ensure the asymptotic variances exist, and imposing those would make the assumptions less transparent.}
\begin{theorem} \label{thm:2}
Under Assumptions \ref{asn:1}-\ref{asn:2}, the LP and SVAR estimators are jointly asymptotically normal:
    \begin{align}
        \sqrt{T}\begin{pmatrix}
            \hat{\beta}_h^{LP} - \beta_h\\
            \hat{\beta}_h^{VAR} - \beta_h
        \end{pmatrix} \xrightarrow{d} \mathcal{N} \left( \begin{pmatrix}
            0\\
            \textup{aBias}_h
        \end{pmatrix}, \begin{pmatrix}
            \Sigma_h^{LP} & \Sigma_h^{COV} \\
            \Sigma_h^{COV} & \Sigma_h^{VAR}
        \end{pmatrix} \right),
        \label{eq:joint_normality}
    \end{align}
    where $ \Sigma_h^{LP}$ and $\Sigma_h^{VAR}$ are the limiting asymptotic variances of LP and VAR described in \eqref{eq:varianceLP} and \eqref{eq:varianceVAR} in Appendix~\ref{app:proofs}, and where $\Sigma_h^{COV}$ denotes the asymptotic covariance between LP and VAR.
\end{theorem}
Note that $\Sigma_h^{LP}$, $\Sigma_h^{COV}$ and $\Sigma_h^{VAR}$ are given in Corollary A.2 of \cite{olea2024double} for i.i.d. structural errors, and they show that when the number of lags used in LP and VAR is the same, then  $\Sigma_h^{VAR} = \Sigma_h^{COV}$. However, in this paper we allow for more lags in LP than in VAR, and conditional heteroskedasticity, in which case their exact form depends on the data generating process. Nevertheless, they do not depend on the misspecification when $\zeta=1/2$, and will be estimated by the bootstrap.

\vspace{1em}
\textbf{Bias - variance trade-off.}
The asymptotic properties of LP and VAR suggest a natural trade-off. In terms of bias, LP is consistent under local misspecification, while VAR exhibits an asymptotic bias of order $T^{-\zeta}$. In terms of variance, the LP variance can increase with the horizon due to the accumulated shocks in $u^h_{i,t+h}$ obtained by backward substituting the VAR model in \eqref{DGP}.  On the other hand, under no misspecification and with i.i.d. normal structural shocks, the VAR estimator is the maximum likelihood estimator and is therefore more efficient than LP. Theorem \ref{thm:2} shows that the local misspecification shifts the VAR distribution but the asymptotic variance stays the same. Under conditional heteroskedasticity, it is not known how the VAR and LP variances compare, but it is likely that VAR still has lower variance at higher horizons, as this variance eventually declines under stationarity because $A^h \to 0$ as $h \to \infty$, while this is not the case for LP.  

These observations motivate targeting LP towards the VAR: the variance of LP can be reduced by shrinking 
toward VAR, accepting a small amount of bias in exchange for lower variance. Additionally, \cite{ludwig2024local} shows that LPs with $p$ lags can be represented as combinations of equal and higher order VARs up to order $p+h-1$, and therefore that, for a given number of lags, the LPs use a more complex model compared to VARs of equal order $p$ or lower. This further motivates targeting LP toward the VAR.

\subsection{Targeted Local Projections}
\label{sec:TLP}

We propose a targeted local projections (TLP) estimator that shrinks LP impulse 
responses toward their VAR counterparts at each horizon. This subsection defines 
the estimator, derives its closed-form solution, and discusses its relationship 
to existing methods.

\vspace{1em}
\textbf{Estimator.}
The TLP estimator $\hat{\beta}_h^{\text{TLP}}$ minimizes a penalized LP
objective function at each horizon over $\tilde \beta_h$:
\begin{align}
    Q_n(\tilde \beta_h) = \big( Y_{h} - X \tilde \beta_h \big)^2 
    + \lambda_h \big( \tilde \beta_h - \hat{\beta}_h^{VAR} \big)^2,
    \label{eq:TLP_objective}
\end{align}
where $\lambda_h > 0$ controls the degree of shrinkage. 
Solving for $\tilde \beta_h$ yields:
\begin{align}
    \hat{\beta}_h^{\text{TLP}} = (X'X + \lambda_h)^{-1} 
    (X'Y_{h} + \lambda_h \hat{\beta}_h^{VAR}),
    \label{eq:TLP_estimator}
\end{align}
where $X'X$ is a scalar after projecting out controls. This expression can be 
rewritten as a linear combination of LP and VAR estimators:
\begin{align}
    \hat{\beta}_h^{\text{TLP}} 
    &= (X'X + \lambda_h)^{-1} X'Y_{h} 
       + (X'X + \lambda_h)^{-1} \lambda_h \hat{\beta}_h^{VAR} \nonumber\\
    &= (X'X + \lambda_h)^{-1} X'X \, \hat{\beta}_h^{LP} 
       + \big(1 - (X'X + \lambda_h)^{-1} X'X\big) \hat{\beta}_h^{VAR} \nonumber\\
    &= v(\lambda_h) \hat{\beta}_h^{LP} 
       + (1 - v(\lambda_h)) \hat{\beta}_h^{VAR},
    \label{eq:TLP_convex}
\end{align}
where $v(\lambda_h) = (X'X + \lambda_h)^{-1} X'X$ is a scalar weight and is by construction in the interval $(0,1)$. The TLP 
impulse response is therefore a linear combination of LP and VAR. Note that when $\lambda_h \rightarrow 0$, 
the weight $v(\lambda_h) \rightarrow 1$ and TLP reduces to LP, while when $\lambda_h \to \infty$, 
the weight $v(\lambda_h) \to 0$ and TLP reduces to VAR.

\vspace{1em}
\textbf{Horizon-specific weights.}
Unlike model averaging approaches that apply uniform weights across horizons 
\citep{olea2024double}, TLP allows the shrinkage parameter $\lambda_h$ to vary 
with the horizon. This flexibility is desirable because the relative performance 
of LP and VAR differs across horizons. At short horizons, LP estimates have low 
bias and are therefore more informative. At higher horizons, LP variance increases 
due to accumulated shocks in the regression error, making it preferable to place 
greater weight on VAR. Moreover, \cite{ludwig2024local} shows that the complexity of LPs relative to VARs of equal or lower order increases with the horizon, further motivating horizon-specific shrinkage. 

\subsection{Comparison to Other Estimators}
\textbf{Bayesian local projections (BLP). }In the case of no additional controls and for a given $\lambda_h$, the mean of the TLP estimator is equal
to the BLP posterior mean in \cite{ferreira2023bayesian}. Despite 
this connection, three distinctions are worth noting: BLP assumes the VAR is correctly specified; it does not project out controls, so the BLP impulse 
    responses are not necessarily linear combinations of LP and VAR impulse responses; it  relies on a presample VAR, while TLP uses the full sample.

\textbf{Smooth local projections (SLP). } TLP and SLP
\citep{barnichon2019impulse, plagborg2016essays} both employ ridge-type regularization 
to reduce the variance of impulse response estimates. However, they differ in 
their regularization targets and in how shrinkage operates across horizons. This 
subsection describes the SLP estimator and highlights these distinctions. As with 
TLP, we project out control variables before estimation to avoid targeting 
nuisance parameters.

Stack 
the projected regressor $X$ and outcome $Y_h$ across horizons $h = 0, \ldots, H$:
\begin{align}
    \tilde{X} &= I_{H} \otimes X, \qquad 
    \tilde{Y} = \begin{pmatrix} Y_0 \\ \vdots \\ Y_{H} \end{pmatrix}, \qquad
    \tilde{\epsilon} = \begin{pmatrix} \epsilon_0 \\ \vdots \\ \epsilon_{H} 
    \end{pmatrix},
\end{align}
so that $\tilde{Y} = \tilde{X} \beta + \tilde{\epsilon}$, where 
$\beta = (\beta_0, \ldots, \beta_{H})'$ collects the $H+1$ impulse responses 
across all horizons. The simplest version of SLP penalizes second differences of adjacent impulse responses. 
Define the second-order difference matrix:
\begin{align}
    L = 
    \begin{bmatrix}
        1 & -2 & 1 & 0 & \cdots & 0 \\
        0 & 1 & -2 & 1 & \cdots & 0 \\
        \vdots & \vdots & \ddots & \ddots & \ddots & \vdots \\
        0 & \cdots & 0 & 1 & -2 & 1 
    \end{bmatrix}_{(H+1-2) \times (H+1)}.
\end{align}
The SLP estimator minimizes the penalized objective:
\begin{align}
    Q_n(\beta) = \| \tilde{Y} - \tilde{X} \beta \|^2 
    + \tilde{\lambda} \| L \beta \|^2,
    \label{eq:SLP_objective}
\end{align}
where $\tilde{\lambda} \geq 0$ controls the degree of smoothing. Setting 
$P = L'L$, we obtain:
\begin{align}
    \hat{\beta}^{\text{SLP}} = 
    (\tilde{X}' \tilde{X} + \tilde{\lambda} P)^{-1} \tilde{X}' \tilde{Y}.
    \label{eq:SLP_estimator}
\end{align}
In the limit $\tilde{\lambda} \to \infty$, the penalty forces $L\beta \to 0$, 
which shrinks the impulse response function toward a straight line. 
Conversely, when $\tilde{\lambda} = 0$, SLP reduces to LP. 
For intermediate values, SLP smooths adjacent horizons toward each other. 

\vspace{1em}
\textbf{Selection of the SLP smoothing parameter.}
Two approaches have been proposed for selecting $\tilde{\lambda}$:
\begin{enumerate}
    \item[(i)] \textit{Cross-validation.} \cite{barnichon2019impulse} recommend 
    $k$-fold cross-validation based on \cite{racine1997feasible}. In their 
    empirical application, they use leave-one-out cross-validation with the 
    standard shortcut formula that avoids refitting for each held-out observation.
    
    \item[(ii)] \textit{Unbiased risk estimator.} \cite{plagborg2016essays} 
    proposes minimizing an unbiased risk estimator (URE):
    \begin{align}
        \min_{\tilde{\lambda}} \hat{R}(\tilde{\lambda}) = 
        T \| \hat{\beta}^{\text{SLP}} - \hat{\beta}^{LP} \|^2_W 
        + 2 \operatorname{tr}\{ W \psi(\tilde{\lambda}) \hat{\Sigma} \},
        \label{eq:URE}
    \end{align}
    where $\hat \beta^{LP} = (\hat \beta_0^{LP}, \ldots, \hat \beta_H^{LP})'$, $|| a||_W^2 = a'W a$, $W$ is a weighting matrix (typically $W = I_H$), 
    $\psi(\tilde{\lambda}) = (\tilde{X}' \tilde{X} + \tilde{\lambda} P)^{-1} 
    \tilde{X}' \tilde{X}$, and $\hat{\Sigma}$ is a HAC estimator of the long-run 
    covariance matrix of LP that accounts for serial correlation across horizons.
\end{enumerate}
In unreported simulations, we find that both criteria yield similar results. For 
comparability with TLP, we select $\tilde{\lambda}$ by minimizing the risk 
criterion~\eqref{eq:URE} throughout. This criterion is analogous to the risk 
function used for TLP and discussed in section~\ref{sec:select_lambda}.

\vspace{1em}
\textbf{Key differences between SLP and TLP.}
Three features distinguish TLP from SLP. First, SLP smooths across horizons, while TLP shrinks toward the VAR-implied impulse responses. Second, by construction, SLP applies a single 
    smoothing parameter $\tilde{\lambda}$ uniformly across all horizons, while TLP 
    allows the shrinkage parameter $\lambda_h$ to vary with the horizon, 
    reflecting the different bias - variance trade-offs at short and long horizons. 
    As the variance of LP can increase with the horizon, while the variance of VAR decreases with the horizon and the asymptotic bias of VAR remains the same across horizons, horizon-specific weights are desirable. Third, SLP smooths adjacent horizons toward 
    each other, which creates distortions at boundary horizons.

\section{Selecting the Shrinkage Parameter}
\label{sec:select_lambda}

The TLP estimator depends on the shrinkage parameter $\lambda_h$, which governs 
the weight placed on LP versus VAR at each horizon. A key advantage of TLP is 
that the optimal weight we propose here has a closed-form solution and can be computed directly 
from the data. This section derives this solution by minimizing the mean-squared error of TLP and is related to the focused information criterion 
of \cite{claeskens2003focused}.

\vspace{1em}
\textbf{Objective function.}
We select the shrinkage parameter at each horizon by minimizing an asymptotically unbiased
estimator of the mean-squared error:
\begin{align}
\, R_h(\lambda_h) = 
    T \, \mathbb{E} \big[ \big( \hat{\beta}_h^{\text{TLP}} - \beta_h \big)^2 \big],
    \label{eq:risk_min}
\end{align}
where recall that $\beta_h$ is the true impulse response. The risk decomposes into squared 
bias and variance:
\begin{align}
    R_h(\lambda_h) = \underbrace{T \big( \mathbb{E} [\hat{\beta}_h^{\text{TLP}}(\lambda_h) 
    - \beta_h ] \big)^2}_{\text{bias}^2(\lambda_h)} + \underbrace{T \, \mathbb{E} 
    \big[ \big( \hat{\beta}_h^{\text{TLP}}(\lambda_h) - \mathbb{E}[\hat{\beta}_h^{\text{TLP}}
    (\lambda_h)] \big)^2 \big]}_{Var(\lambda_h)}.
    \label{eq:bias_var_decomp}
\end{align}
 Assume we have consistent estimators $\hat{\Sigma}^{\mu}_h \stackrel{p}{\rightarrow} \Sigma^{\mu}_h, \text{ for } \mu \in \{LP,VAR,COV\}$. Consistent estimators for LP can be obtained by standard HAC estimators; however, we use the bootstrap. Consistent estimators for the rest of the quantities are discussed in the next subsection. 

\begin{theorem} \label{theo:3}
    \underline{Asymptotically unbiased risk criterion}. Under Assumptions \ref{asn:1}-\ref{asn:2}, for $\zeta=1/2$, the feasible risk criterion: 
\begin{align}
    \hat{R}_h(\lambda_h)
    & = (1 - v(\lambda_h))^2
       T (\hat{\beta}_h^{LP} - \hat{\beta}_h^{VAR})^2 + (2 v(\lambda_h)-1) \hat \Sigma^{LP}_h + 2(1 - v(\lambda_h)) \hat \Sigma^{COV}_h.
    \label{eq:TLP_risk}
\end{align}     
is asymptotically unbiased: $\sup_{\lambda_h \in (0, \infty)} |\E[\hat R_h(\lambda_h)] -R_h(\lambda_h)| \rightarrow 0$.
\end{theorem}

Theorem~\ref{theo:3} derives an asymptotically unbiased risk criterion under local misspecification of order $T^{-1/2}$, the only one under which a meaningful asymptotic bias - variance trade-off is possible. The criterion~\eqref{eq:TLP_risk} is quadratic in $v(\lambda_h)$, so minimization 
yields a unique closed-form solution. 

Let the optimal shrinkage parameter minimize the feasible risk criterion:
$$
\hat \lambda_h =\mbox{ arg} \min_{\lambda_h \in (0,\infty)} \hat{R}_h(\lambda_h),
$$
when the implied weights $v(\hat \lambda_h)$ are positive, and $\hat \lambda_h =\infty$, so $v(\hat \lambda_h) = 0$, otherwise.
We then obtain the following closed-form expression for the TLP weights: 
\vspace{1em}
\begin{theorem}\label{theo:4}
\underline{Optimal weight}. Under Assumptions \ref{asn:1}-\ref{asn:2}, for $\zeta=1/2$,
\begin{align}
   v(\hat \lambda_h) &= \mbox{max}\left [1- \frac{\hat \Sigma_h^{LP} -\hat \Sigma_h^{COV}}{T(\hat \beta_h^{LP}-\hat \beta_h^{VAR})^2},0\right].
    \label{eq:weight}
\end{align}
If, furthermore, the LP and VAR have the same number of lags, and the structural shocks are i.i.d., then 
\begin{align}
   v(\hat \lambda_h) &= \mbox{max}\left [1- \frac{\hat \Sigma_h^{LP} -\hat \Sigma_h^{VAR}}{T(\hat \beta_h^{LP}-\hat \beta_h^{VAR})^2},0\right].
    \label{eq:weight_simple}
\end{align}
\end{theorem}

The result \eqref{eq:weight_simple} follows from Corollary A.2 of \cite{olea2024double}, which shows that with the same number of lags and i.i.d. structural shocks, $\Sigma^{VAR}_h = \Sigma^{COV}_h$. It is worth noting:

\begin{enumerate}
     \item[(i)] If $1/4<\zeta<1/2$, the bias term dominates the variance terms, the weight $v(\hat{\lambda}_h)$ approaches one, and TLP places all weight on LP, avoiding the biased VAR.
     
    \item[(ii)] If $\zeta>1/2$, the VAR has no asymptotic bias, the limit expectation of ${T(\hat \beta_h^{LP}-\hat \beta_h^{VAR})^2}$ is $\hat \Sigma_h^{LP} +\hat \Sigma_h^{VAR}-2\hat \Sigma_h^{COV}$, which is equal to the numerator in \eqref{eq:weight_simple} $\hat \Sigma_h^{LP} -\hat \Sigma_h^{COV}$ if $\hat \Sigma_h^{VAR}=\hat \Sigma_h^{COV}$. However, even in the latter case, because the weights are random, not all weight will be placed on VAR since the maximum is with positive probability not achieved at $0$.
    
    \item[(iii)] If $0<\zeta\leq 1/4$, a case not covered by our assumptions, both LPs and VAR have asymptotic bias, of a larger order than the variance terms, and TLP asymptotically places all weight on LP.
\end{enumerate}

\vspace{1em}
\textbf{Comparison to uniform model averaging.}
An alternative to minimizing risk is to choose weights based on theoretical 
misspecification bounds. \cite{olea2024double} propose selecting the weight 
based on $M = \lVert \alpha(L) \rVert$, the norm of the misspecification 
polynomial, which applies to all horizons uniformly. Their Corollary 4.2 states 
that the optimal weight $\omega$ between LP and VAR should be:
\begin{align}
    \omega = \frac{M^2}{M^2+1}.
    \label{eq:weight_MA}
\end{align}
Comparing equations~\eqref{eq:weight} and~\eqref{eq:weight_MA},  it can be noted that for large misspecifications $0<\zeta<1/2$, TLP asymptotically places all weight on LP. 

\textbf{Comparison to \cite{chen2026estimator}.} In their asymptotic theory, \cite{chen2026estimator} consider averaging LP and VAR when they both do not exhibit asymptotic bias. Because \cite{chen2026estimator} do not shrink LP towards the VAR, but seek the linear combination of LP and VAR that minimizes the mean-squared error of their estimator, their formula is different from ours under a correctly specified VAR.\footnote{Note that while their theory is derived under correct specification, their estimator can also be applied to misspecified cases, and their simulations are under less severe misspecifications than ours.} 
\section{Variance estimation \label{sec:variance}}

Constructing confidence intervals for impulse response functions requires reliable 
variance estimates. For LP, standard variance estimators are available. For VAR 
impulse responses, however, it is known that the delta method
performs poorly in finite samples. TLP inherits these challenges since it combines 
LP with VAR estimates. We therefore employ the bootstrap.

We propose an MSDB procedure for inference on LP, VAR, SLP, 
and TLP estimators. The double bootstrap is employed to calculate \textit{studentized} confidence intervals, which are important to achieve asymptotic refinements when the misspecification is not too large (when $\zeta>1/2$).\footnote{Recently, \cite{cavaliere2024bootstrap} have demonstrated that in the presence of asymptotic bias, one can use either pre-pivoting or double bootstrap to restore coverage. Their setting does not apply to our bootstrap because neither the single nor the double bootstrap we employ can replicate the asymptotic bias.} As discussed earlier, for $\zeta=1/2$, the VAR and TLP estimators suffer from an asymptotic bias. The residual bootstrap we employ does not replicate this bias as it is centered around the VAR and TLP estimators, but it does replicate their variance. This means that the bootstrap t-statistics used for studentization are centered around zero, and the confidence intervals are shifted due to asymptotic bias, distorting coverage. The benefit of the bootstrap in this case is that by centering around zero, these confidence intervals are not further shifted. 

Our procedure shares similarities with the VAR-based resampling scheme in
\cite{montiel2021local}, \cite{olea2024double} and \cite{chen2026estimator}, but differs in several aspects:
\begin{enumerate}
    \item[(i)] We use the recursive residual moving block 
    bootstrap \citep{bruggemann2016inference}, while \cite{olea2024double} use a recursive residual i.i.d. bootstrap, and \cite{chen2026estimator} propose a residual i.i.d. bootstrap after employing a sieve autoregression. All of these are valid under no misspecification.
       
    \item[(ii)] We employ a double bootstrap rather than a single 
    bootstrap. A single bootstrap does not provide reliable variance estimates needed for 
    studentization, since the delta method performs poorly in small 
    samples. The double bootstrap provides consistent 
    variance estimates for standardization; intuitively, the second layer of bootstrap is generated with the true DGP being the first level bootstrap, so it has no bias relative to that DGP, rendering consistent estimates of the VAR variances and the covariance between LP and VAR.
    
    \item[(iii)] When constructing t-statistics, we subtract the bootstrap mean 
    rather than the VAR-implied impulse response. This mean-centering removes finite sample bias in the numerator of the bootstrap 
    t-statistics. 
    
    \item[(iv)] We use symmetrized confidence intervals \citep{hall1988symmetric}. Asymmetric intervals work well for LP, but they do not account for the skewness in the VAR and TLP t-statistics. 
\end{enumerate}

\vspace{1em}
\textbf{First level bootstrap.}
\cite{bruggemann2016inference} showed that if Assumptions similar to \ref{asn:1}-\ref{asn:2} hold with $\alpha(L)=0$, so the model is correctly specified, the recursive moving block bootstrap is valid under conditional heteroskedasticity of unknown form for the VAR estimates. This bootstrap is the one we employ and it proceeds in the following steps: 
(i) estimate the reduced-form VAR and obtain the residuals $\hat{u}_t$; (ii) divide the residual series into overlapping blocks of length 
    $\ell = T^{1/3}$, following \cite{gonccalves2004bootstrapping}; (iii) draw blocks with replacement and paste them together until the series $\{u_t^{(b)}\}_{t=1}^T$ reaches length $T$; (iv) reconstruct the bootstrap sample recursively
$ 
        y_t^{(b)} = \hat{A} y_{t-1}^{(b)} + u_t^{(b)}
$; (v) in each bootstrap replication, estimate impulse responses using both VAR and LP.

\vspace{1em}
\textbf{Second level bootstrap}
The double bootstrap proceeds as follows:
\begin{enumerate}[i)]
    \item Treat the original sample as $b_1 = 1$ and generate $B_1 - 1$ bootstrap 
    samples using the moving-block procedure described above. For each sample
    $b_1 = 1, \dots, B_1$, estimate impulse responses $\hat{\beta}_{h,b_1}^{\mu}$ 
    at each horizon $h$ for $\mu \in \{LP, VAR, SLP\}$.
    
    \item For each first-level bootstrap sample, apply the same moving-block resampling 
    procedure to generate $B_2$ second-level bootstrap replications, now treating 
    the first-level sample as the data. This yields 
    $\hat{\beta}_{h,b_1,b_2}^{\mu}$ for $b_2 = 1, \dots, B_2$.
    
    \item Compute the variances and the covariance between LP and VAR for $b_1>1$ from second-level replications:
    \begin{align}
        \hat{\Sigma}^\mu_{h,b_1} = \widehat{\text{Var}}_{b_2}
            \left(\hat{\beta}_{h,b_1,b_2}^{\mu}\right), \qquad
        \hat{\Sigma}^{COV}_{h,b_1} = \widehat{\text{Cov}}_{b_2}
            \left(\hat{\beta}_{h,b_1,b_2}^{LP}, \hat{\beta}_{h,b_1,b_2}^{VAR}\right).
    \end{align}
\end{enumerate}
For LP and VAR, these quantities are sufficient to construct the t-statistics. 
TLP is constructed from LP and VAR using the weights described below. SLP 
requires additional steps to obtain variance estimates. 

\textbf{TLP.} To compute the weights $v(\hat{\lambda}_h)$, we need consistent 
estimates of $\hat{\Sigma}_h^{LP}$, $\hat{\Sigma}_h^{VAR}$, and 
$\hat{\Sigma}_h^{COV}$. We obtain these by averaging the double-bootstrap estimates across 
first-level bootstrap samples:
\begin{align}
    \hat{\Sigma}_h^{\mu} &= \frac{1}{B_1} \sum^{B_1}_{b_1=1} 
        \hat{\Sigma}_{h,b_1}^{\mu}, \quad \mu \in \{LP, VAR\}, \qquad
    \hat{\Sigma}_h^{COV} = \frac{1}{B_1} \sum^{B_1}_{b_1=1} 
        \hat{\Sigma}_{h,b_1}^{COV}.
\end{align}
These are used to compute the weights $v(\hat{\lambda}_{h,b_1})$ using 
equation~(\ref{eq:weight}), with first level bootstrap used for LP and VAR estimates. We 
re-estimate the bias components of the weights of TLP in each first level bootstrap replication. The bootstrap replications of TLP are then
\begin{equation}
    \hat{\beta}_{h,b_1}^{TLP}(\hat{\lambda}_{h,b_1}) = v(\hat{\lambda}_{h,b_1})
        \hat{\beta}_{h,b_1}^{LP} + (1 - v(\hat{\lambda}_{h,b_1}) )
        \hat{\beta}_{h,b_1}^{VAR},
\end{equation}
and the variance of each first-level replication is computed as the average variance of the TLP estimates across the second level bootstrap sample.

\textbf{SLP.} In principle, it is possible to also estimate the SLP variance using the 
second bootstrap layer. However, to reduce computational burden, we instead calculate the variance only in the first-level loop, using the 
LP variance estimates from the second-level bootstrap:
\begin{align}
    \hat{\Sigma}_{b_1}^{SLP} = (\tilde{X}'\tilde{X} + \tilde{\lambda} P)^{-1} 
        (\tilde{X}'\tilde{X}) \, \hat{\Omega}_{b_1}^{LP} \, 
        (\tilde{X}'\tilde{X}) (\tilde{X}'\tilde{X} + \tilde{\lambda} P)^{-1},
\end{align}
where $\hat{\Omega}_{b_1}^{LP}$ is the LP variance-covariance matrix across all horizons, estimated by stacking second level bootstrap draws of LP and taking covariance of the resulting matrix, 
and $\tilde{\lambda}$ is fixed at its original sample estimate. This approach avoids the 
$B_2$ second-level iterations for SLP while still using the double-bootstrap 
variance estimates for LP. The result is an $H \times H$ covariance matrix; the 
diagonal element $[\hat{\Sigma}_{b_1}^{SLP}]_{hh}$ gives the variance at 
horizon $h$.

\textbf{Constructing t-statistics.} 
For all estimators $\mu \in \{LP, VAR, SLP, TLP\}$, we construct t-statistics by 
centering around the bootstrap mean:
\begin{align}
    t^\mu_{h,b_1} = \frac{\hat{\beta}_{h,b_1}^\mu - \bar{\beta}_h^\mu}
        {\sqrt{\hat{\Sigma}_{h,b_1}^\mu}}, \quad \text{where} \quad 
    \bar{\beta}_h^\mu = \frac{1}{B_1}\sum_{b_1=1}^{B_1}\hat{\beta}_{h,b_1}^\mu.
\end{align}
Subtracting the bootstrap mean rather than the VAR-implied impulse response 
removes finite sample bias in the numerator of the t-statistics.

\textbf{Symmetric confidence intervals.}
The VAR bias under misspecification shifts the location of its confidence interval and 
leads to asymmetric t-statistic distributions. TLP inherits some of this asymmetry by 
construction. Additionally, the exponentiation in VAR exacerbates the issues. Symmetric confidence intervals \citep{hall1988symmetric} alleviate 
this problem and, as discussed in \cite{hall1988symmetric}, tend to exhibit faster convergence rates than their asymmetric 
counterparts.

For any estimator $\mu \in \{LP, VAR, SLP, TLP\}$, we rank the bootstrap 
t-statistics by absolute value and take the $(1-\alpha)$-quantile of 
$|t^\mu_{h,b_1}|$ as the critical value $\hat{t}^\mu_{h,1-\alpha}$. The confidence 
interval is:
\begin{align}
    \left(\hat{\beta}_h^\mu - \hat{t}^\mu_{h,1-\alpha} 
        \sqrt{\hat{\Sigma}_{h,1}^{\mu}}, \quad 
    \hat{\beta}_h^\mu + \hat{t}^\mu_{h,1-\alpha} 
        \sqrt{\hat{\Sigma}_{h,1}^{\mu}} \right),
\end{align}
where $\hat{\Sigma}_{h,1}^{\mu}$ is the variance estimate from the original 
sample. 

\section{Simulation Study \label{sec:simulation}}

We conduct 1,000 repetitions for sample sizes $T \in \{200,800\}$, comparing the performance of LP, SLP, TLP, VAR, and BLP estimators. The MSDB uses $B_1 = 200$ first level replications and $B_2= 100$ second-level replications. The LP, SLP, and TLP estimators are implemented using $p=10$ lags on the right-hand side, whereas BLP and VAR, as well as the target model for TLP, are specified with either $q \in \{1,8\}$ lags. The target coverage level is set at 90\%. The BLP method utilizes one-sixth of the available data to construct the presample VAR prior, and the coverage is assessed using the posterior distribution.

\textbf{VARMA(1,$\boldsymbol{\infty}$).} The data is generated by the following model
\begin{align}
    Y_t = A Y_{t-1} + \Gamma(I + \eta_T\alpha(L))\epsilon_t,
\end{align}
where \begin{align}
    \Gamma^{-1} = 
    \begin{bmatrix}
        1 & 0\\
        -1.07 & 1
    \end{bmatrix} \;\;\;\;\     A = 
    \begin{bmatrix}
        0.2 & 0\\
        -0.764 & 0.985
    \end{bmatrix},
    \end{align}
    and $\epsilon_t \sim \text{i.i.d. } \mathcal N(0,I_2)$.
This DGP is taken from \cite{olea2024double}, which simulates the \cite{smets2007shocks} model, a well-known DSGE model that is rich enough to match the second-moment properties of standard macroeconomic time-series data. This model implies a VARMA representation with interpretable structural shocks, implying that any finite order VAR is misspecified. The difference from the DGP in \cite{olea2024double} is that  we use the scaling $\eta_T=\eta \times \sqrt{\frac{200}{T}}$ to check performance of the double bootstrap as $T$ increases. The multiplication by $\sqrt{200}$ is made such that at $T=200$ we estimate the real DGP when $\eta=1$. The data generating process used is VARMA$(1,100)$, with $\eta \in\{1, 2, 4, 8, 16, 32, 64\}$.

\subsection{Comparing bootstrap variants}

Here, we set $\eta=1$. Figure~\ref{fig:coverage_comparison} compares the MSDB coverage properties against an alternative double bootstrap across four estimators: 
LP, VAR, SLP, and TLP for $\eta=1$. The red solid lines represent MSDB, while the purple
lines show the alternative which is the same double bootstrap but where the t-statistics used in the studentization are
centered around the original VAR estimates, as in 
\cite{montiel2021local}. LP has 10 lags, and SVAR and the target for TLP have $q \in \{1,8\}$ lags.

\textbf{LP and SLP}. For LP, both bootstrap methods achieve coverage close to or above the nominal 90\% level across all horizons and sample sizes, with no difference in the average confidence interval length. The choice of $q$ changes the coverage and the interval length only marginally. For SLP, there is a drop in coverage at the first few horizons. Nevertheless, at higher horizons the centering around the bootstrap mean yields closer to nominal coverage. 

\textbf{VAR and TLP}. For both VAR and TLP, MSDB yields a clear improvement: centering around the VAR estimates leads to substantial undercoverage, consistent with findings in \cite{olea2024double}. MSDB almost restores nominal coverage for both estimators at $T=200$ and $T=800$ for $q=8$. For $q=1$, there is significant undercoverage for VAR and some undercoverage for TLP. Therefore, it is important to use enough lags in the target VAR.

\begin{figure}[H]
    \centering
    \includegraphics[width=\linewidth, height=0.90\textheight, keepaspectratio]{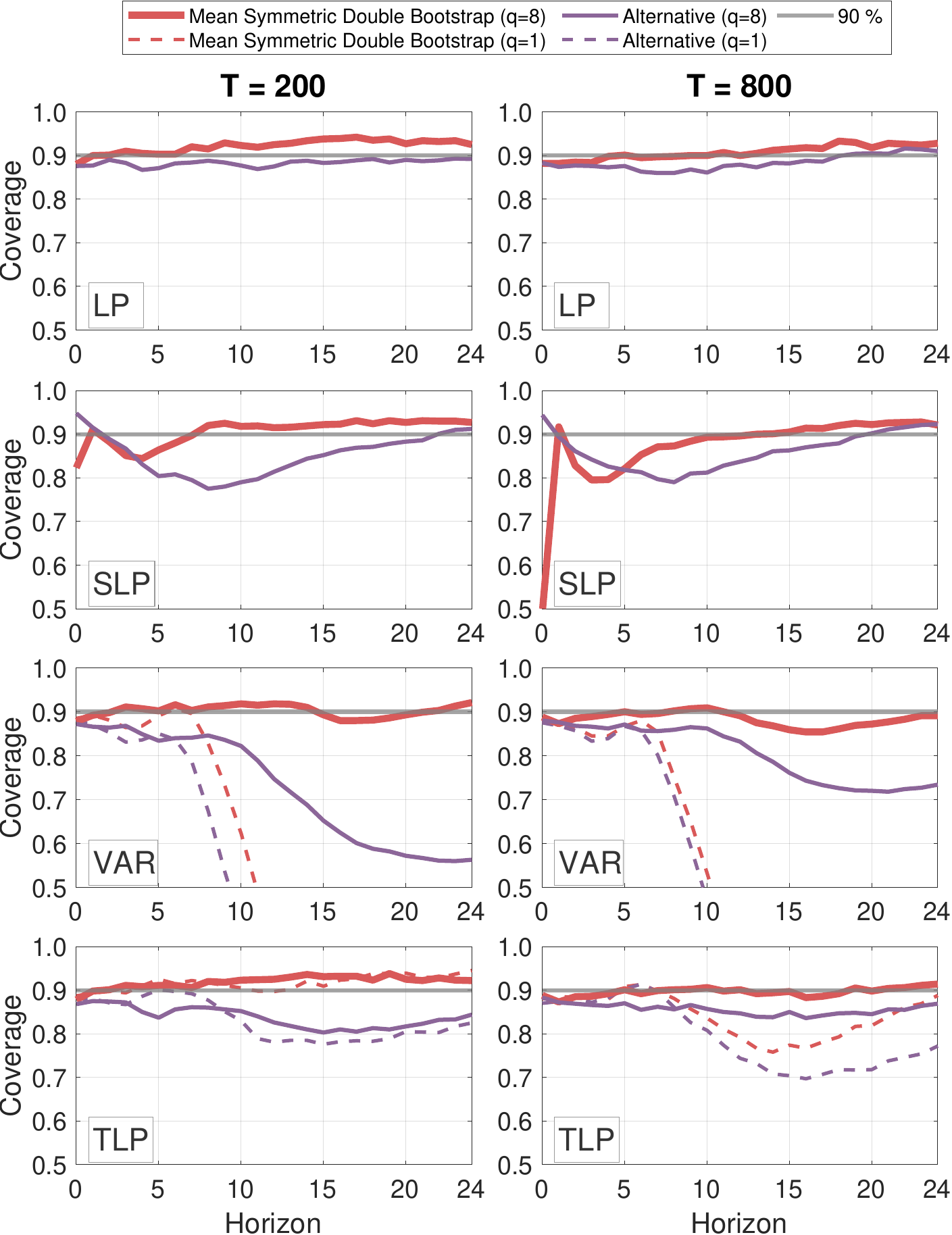}
    \caption{Coverage for $T=200$ (left) and $T=800$ (right), DGP = VARMA(1,$100$), $\eta=1$. Red solid line: MSDB with $q=8$. Purple solid line (alternative method): symmetric double bootstrap centering around VAR estimates. Dashed lines for $q=1$.}
    \label{fig:coverage_comparison}
\end{figure}

\begin{figure}[H]
    \centering
    \includegraphics[width=\linewidth, height=0.90\textheight, keepaspectratio]{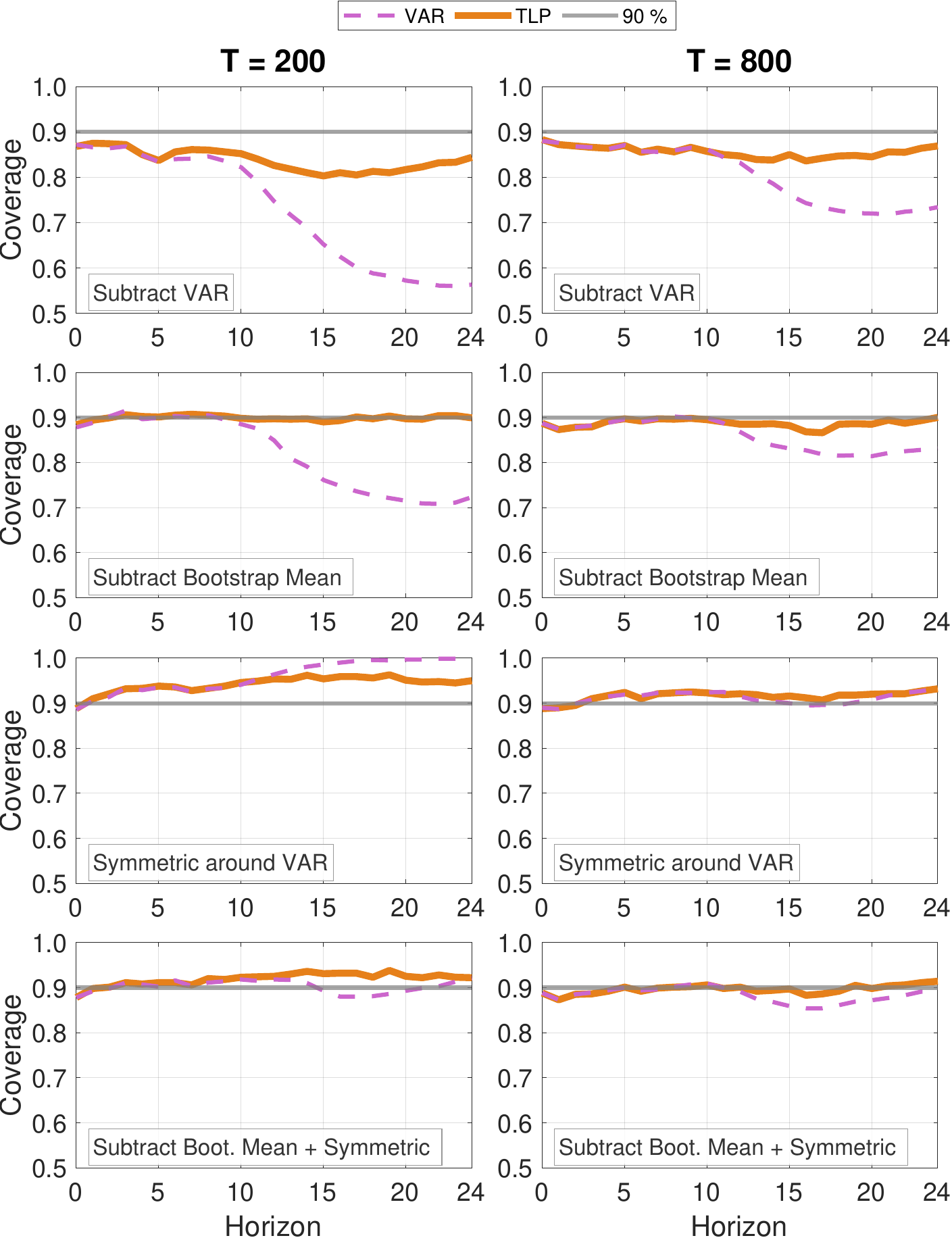}
    \caption{Coverage for $T=200$ (left) and $T=800$ (right), DGP = VARMA(1,$100$), $\eta=1$, $q=8$, Orange solid line: TLP. Pink dashed line: VAR. Rows compare four studentized double bootstraps: Subtract VAR estimate, Subtract Bootstrap Mean, Symmetric around VAR, and Subtract Bootstrap Mean + Symmetric (MSDB).}
    \label{fig:var_tlp_method_coverage}
\end{figure}

\begin{figure}[H]
    \centering
    \includegraphics[width=0.9\linewidth, keepaspectratio]{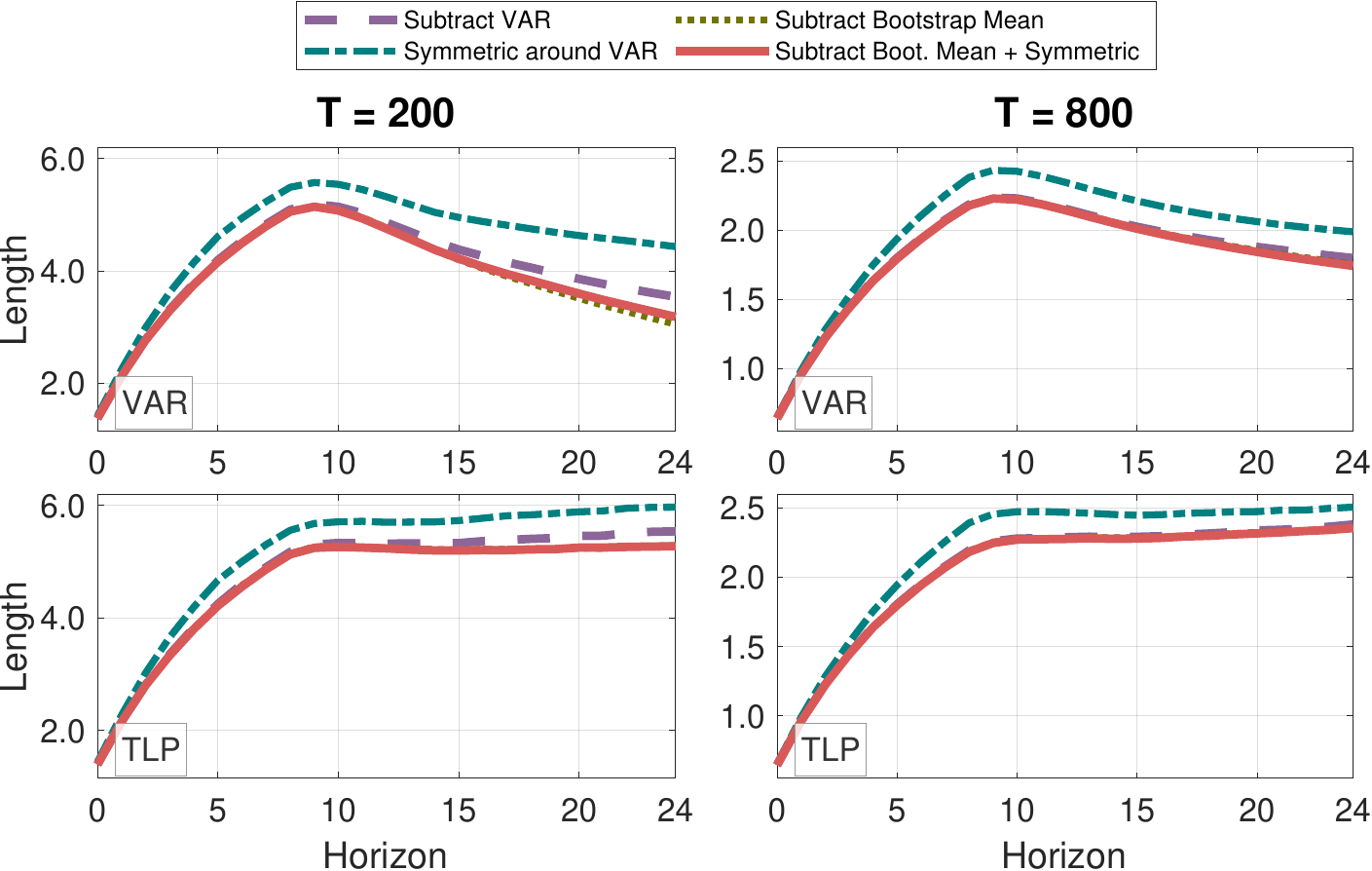}
    \caption{MC Average confidence interval length for VAR and TLP, $T=200$ (left) and $T=800$ (right), DGP = VARMA(1,$100$), $\eta=1$, $q=8$. Each panel compares four studentized double bootstraps: Subtract VAR estimate, Subtract Bootstrap Mean, Symmetric around VAR, and Subtract Bootstrap Mean + Symmetric (MSDB).}
    \label{fig:var_tlp_method_length}
\end{figure}

Figure~\ref{fig:var_tlp_method_coverage} focuses on TLP and VAR with $q=8$, building up MSDB one change at a time. The first row shows significant undercoverage when using a standard double bootstrap centered at the VAR estimate. The second row centers the $t$-statistic numerator around the first level bootstrap mean instead, showing substantial improvements in coverage. The third row keeps VAR centering but switches to symmetric intervals, which pushes coverage, especially for VAR, above the nominal level, while increasing the length of the confidence intervals, as seen in Figure~\ref{fig:var_tlp_method_length}. The fourth row shows the MSDB, which delivers close to nominal coverage for TLP and less distortion for VAR, with the length of the confidence interval remaining approximately the same as in the asymmetric confidence interval cases.

\begin{figure}[H]
    \centering
    \includegraphics[width=0.9\linewidth, keepaspectratio]{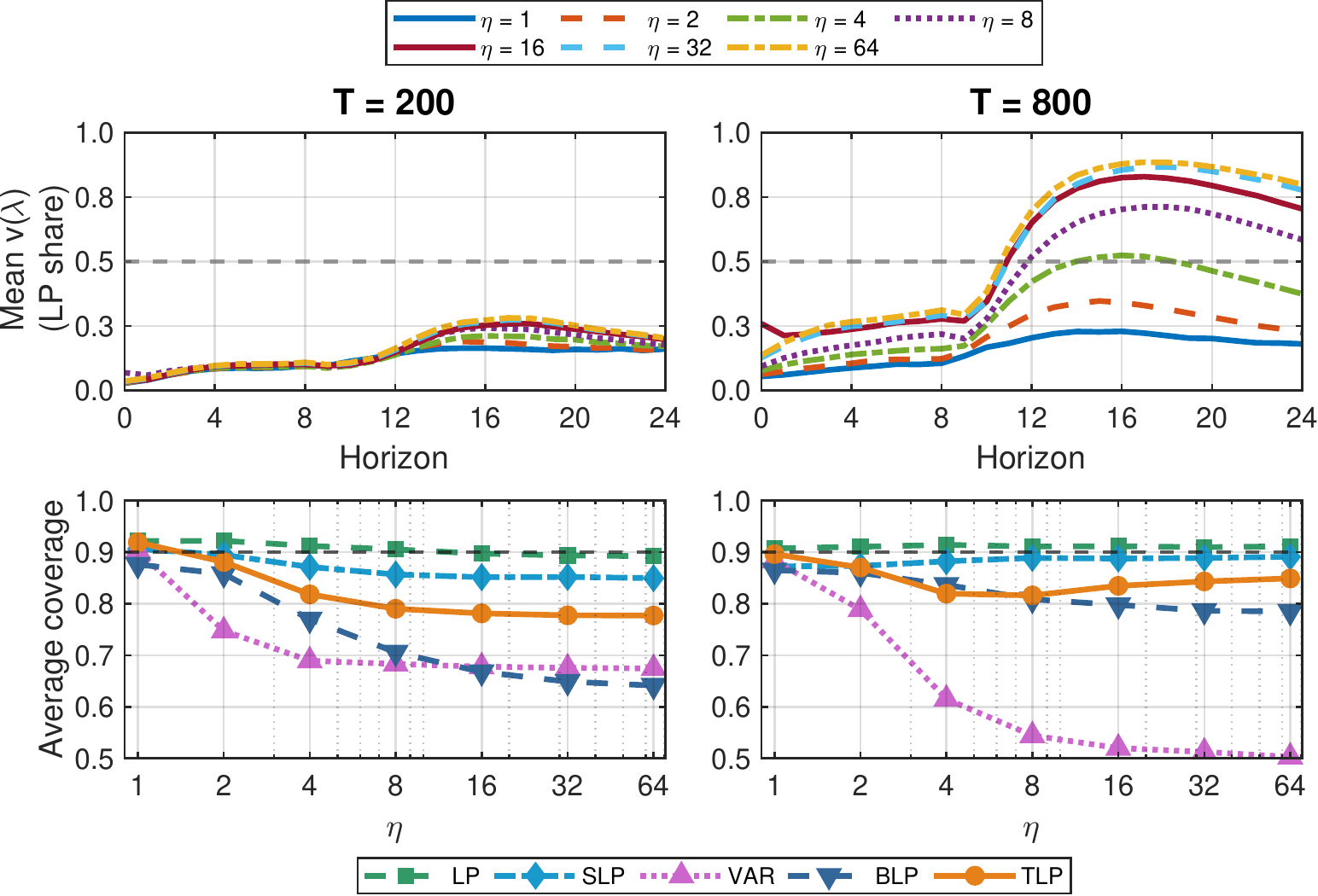}
    \caption{Top row: MC average weight depending on the size of the misspecification $\eta$. Bottom row: MC average coverage depending on the size of the misspecification $\eta$. $T=200$ (left) and $T=800$ (right).}
    \label{fig:mean_weight}
\end{figure}

\subsection{Comparison Across Estimators}

Figure \ref{fig:mean_weight} shows the mean - over 24 horizons - TLP weights and mean coverage attributed to LP across different misspecification magnitudes. As expected, the larger the misspecification, the more weight TLP places on LP, which means that first the TLP coverage deteriorates as the misspecification increases, and then it improves again as TLP approaches LP.

Figures \ref{fig:varma1infty_metrics}-\ref{fig:varma1inftyBB4_metrics} and Appendix Figures \ref{fig:varma1inftyBB8_metrics}-\ref{fig:varma1inftyBB32_metrics} show coverage, 
length, bias, standard deviation, and RMSE for all five estimators, across different misspecification sizes, horizon by horizon. The coverage and length are for MSDB confidence intervals, while the bias, standard deviation and RMSE are Monte Carlo averages computed with the original sample estimates. 

VAR displays significant negative bias at 
longer horizons while LP remains approximately unbiased for large samples until the misspecification reaches $\eta=8$. TLP inherits some 
bias from its VAR target, but this cost is offset by a marked reduction in variance compared to LP. We note that for $\eta=1$, the misspecification is small enough so that the bias - variance trade-off is not present, and VAR has the smallest RMSE. However, even in this case, for most horizons, the RMSE of TLP is similar, and the VAR exhibits slight coverage distortions while the coverage of TLP is close to nominal. As $\eta$ increases, the trade-off becomes visible, and TLP exhibits smaller RMSE than VAR, with marked RMSE improvements over VAR at $\eta =\{2, 4\}$. For the latter misspecifications, the TLP bias is markedly smaller than that of the VAR, and the variance is smaller than that of the LP, no longer increasing with the horizon. As the misspecification increases, some coverage distortions of TLP are to be expected, as discussed in the introduction; however, they are markedly smaller than those of the VAR. Interestingly, as shown in Figure \ref{fig:mean_weight}, these coverage distortions become smaller again as the misspecification further increases and more weight is placed on LP.
 
Since we allow for conditional heteroskedasticity, we also simulate the VARMA(1,$100$) with $\eta=200^{1/2} \times T^{-1/2}$ and the errors following a GARCH(1,1) process: $\epsilon_t = \sigma_t z_t$ with $z_t \sim \text{i.i.d.}N(0,1)$ and $\sigma_t^2 = \omega + \alpha \epsilon_{t-1}^2 + \beta \sigma_{t-1}^2$, where $\omega = 0.05$, $\alpha = 0.10$, and $\beta = 0.85$. The results are similar to Figures~\ref{fig:varma1infty_metrics}-\ref{fig:varma1inftyBB4_metrics} and are shown in Figures~\ref{fig:varma1garch_metrics}~-~\ref{fig:varma1garchBB4_metrics} in Appendix~\ref{app:garch}.

\textbf{SLP and BLP.} SLP achieves nominal coverage at longer horizons but suffers from severe 
undercoverage at short horizons. This pattern is expected, as smoothing requires neighboring estimates, 
which are unavailable at the initial horizon. Even when coverage is near 
nominal, the average length of the CI does not decrease meaningfully 
compared to standard LP. It also exhibits substantial bias at smaller horizons without meaningful reductions in variance relative to LP. BLP also suffers from substantial ``undercoverage'' at short horizons, with wide credible intervals, likely because it uses a sixth of the sample 
as a presample to initialize the prior distribution, reducing effective sample size.

\begin{figure}[H]
    \centering
    \includegraphics[width=\linewidth, height=0.90\textheight, keepaspectratio]{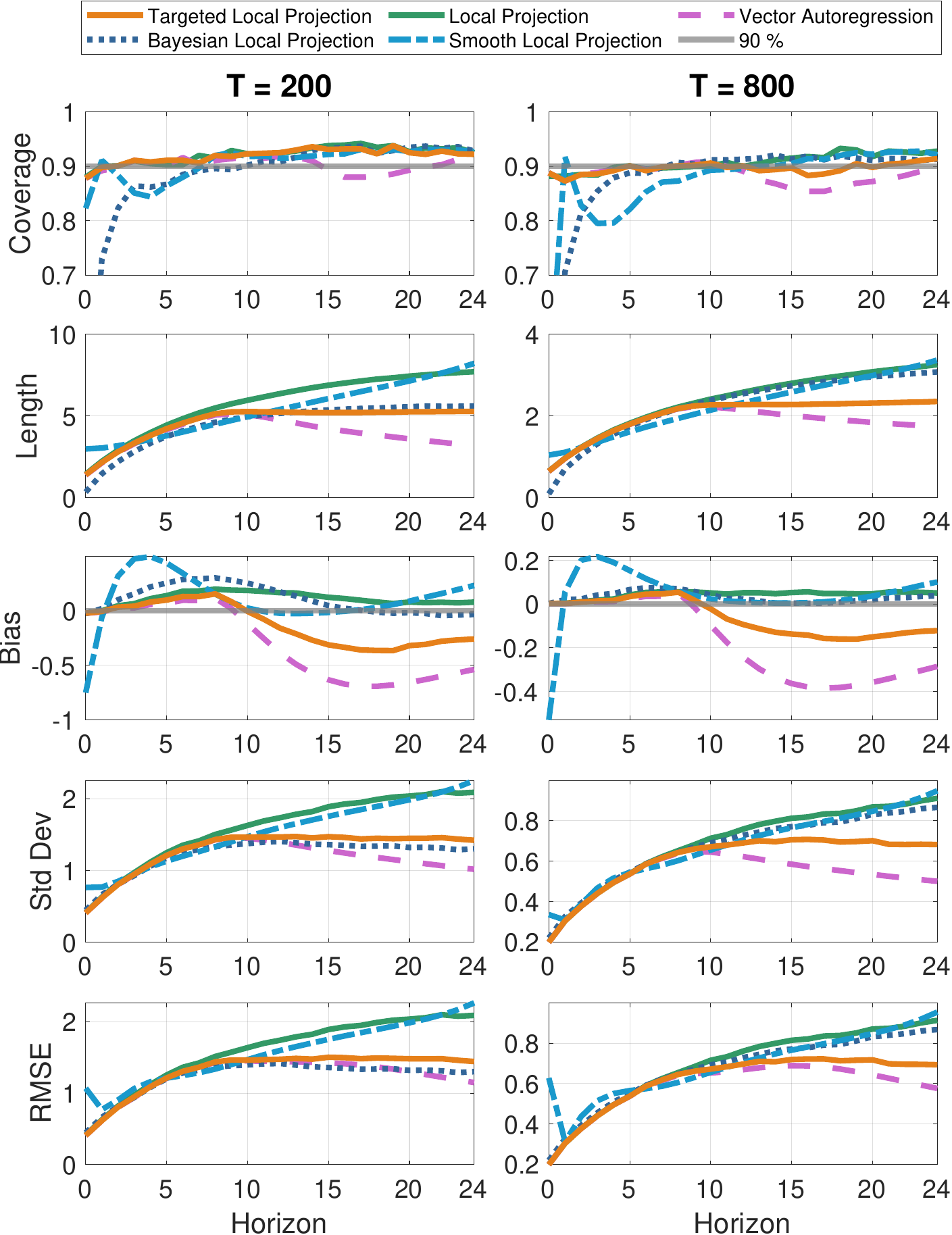}
    \caption{Coverage, length, bias, standard deviation, and RMSE for $T=200$ (left) and $T=800$ (right), DGP = VARMA(1,$100$), $\eta=1$, conditionally homoskedastic errors. Methods: Targeted Local Projections (10,8), Local Projections (10), Vector Autoregression (8), Smooth Local Projections (10) and Bayesian Local Projections (8,8). Inference by MSDB for TLP, LP, VAR and SLP.}
    \label{fig:varma1infty_metrics}
\end{figure}
\newpage

\begin{figure}[H]
    \centering
    \includegraphics[width=\linewidth, height=0.90\textheight, keepaspectratio]{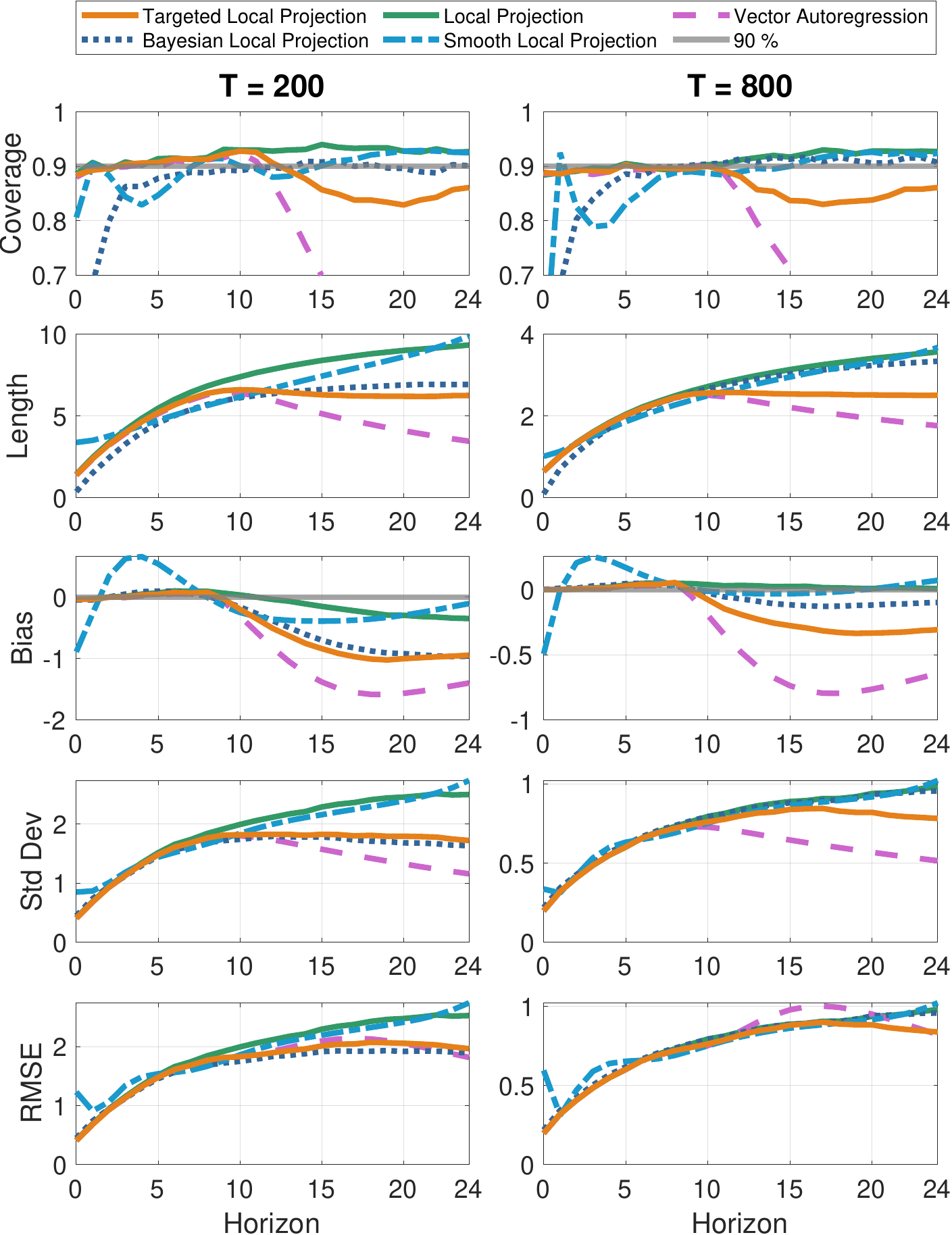}
    \caption{Coverage, length, bias, standard deviation, and RMSE for $T=200$ (left) and $T=800$ (right), DGP = VARMA(1,$100$), $\eta=2$, conditionally homoskedastic errors. Methods: Targeted Local Projections (10,8), Local Projections (10), Vector Autoregression (8), Smooth Local Projections (10) and Bayesian Local Projections (8,8). Inference by MSDB for TLP, LP, VAR and SLP.}
    \label{fig:varma1inftyBB2_metrics}
\end{figure}
\newpage
\begin{figure}[H]
    \centering
    \includegraphics[width=\linewidth, height=0.90\textheight, keepaspectratio]{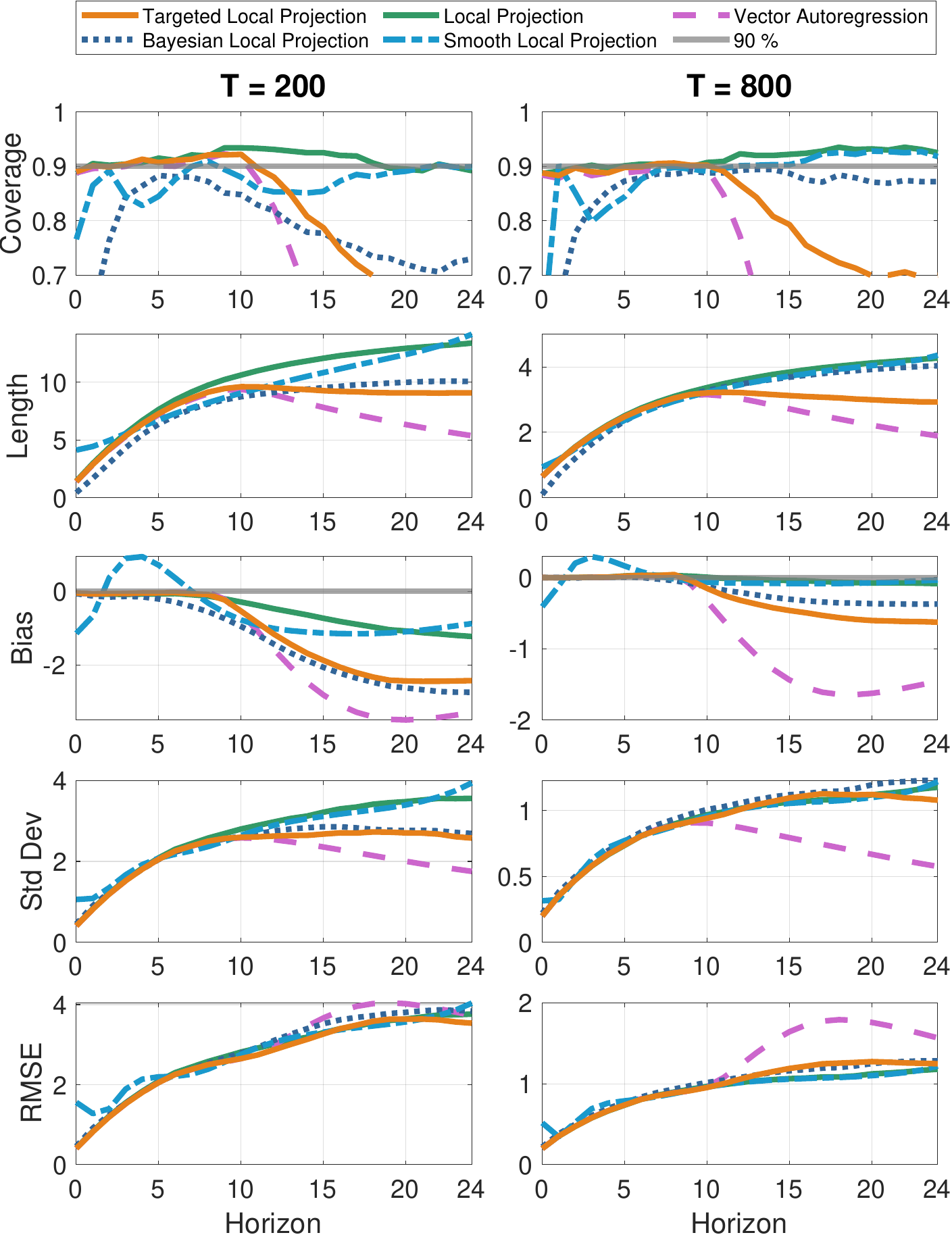}
    \caption{Coverage, length, bias, standard deviation, and RMSE for $T=200$ (left) and $T=800$ (right), DGP = VARMA(1,$100$), $\eta=4$, conditionally homoskedastic errors. Methods: Targeted Local Projections (10,8), Local Projections (10), Vector Autoregression (8), Smooth Local Projections (10) and Bayesian Local Projections (8,8). Inference by MSDB for TLP, LP, VAR and SLP.}
    \label{fig:varma1inftyBB4_metrics}
\end{figure}

\section{Empirical Application}\label{sec:application}
Following \cite{ferreira2023bayesian}, we estimate impulse response functions of key macroeconomic variables to a federal funds rate shock using a seven-variable system that includes real GDP, real consumption, real investment, total hours worked, real wages, the GDP deflator, and the federal funds rate. The data are from the US and span 1954Q3 to 2019Q4, with all variables expressed in logarithms and seasonally adjusted at the annual rate, except for the policy rate. More details on the data construction can be found in \cite{ferreira2023bayesian}, Online Appendix A. 

As in \cite{ferreira2023bayesian}, the monetary policy shock is identified via a recursive ordering, with the federal funds rate placed last. 
We estimate TLP, LP, and VAR impulse responses with $(p,q) = (8,4)$ and $(p,q) = (10,4)$ using the MSDB for inference. The estimated impulse responses are shown in Figure~\ref{fig:fedfundsrate_irfs} for $p=10$, and in Appendix Figure \ref{fig:fedfundsrate_irfs2} for $p=8$. As expected, following a contractionary monetary policy shock, real GDP, consumption, investment, and hours worked decline. From the two figures, it is visible that for most variables except hours worked, the GDP deflator and the federal funds rate, the LP and the VAR estimates are quite different at higher horizons, and a possible reason for these differences could be that the VAR is misspecified at higher horizons, as also suggested in \cite{ferreira2023bayesian}. When the two estimates are very different, the TLP estimates place more weight on LP, and when they are similar, more weight is placed on the VAR. 

The TLP confidence intervals are tighter than their LP counterparts. Tables 1-2 show the horizons at which the impulse responses are statistically significant at the 90\% level for these first four variables, with blue highlighting the horizons at which TLP estimates are significant while LP estimates are not. We see that particularly for real investment, the TLP estimates are larger and remain significant for most horizons, demonstrating the usefulness of our estimator for empirical applications in reducing variance. 

\begin{figure}[H]
\centering
\includegraphics[width=\linewidth, height=0.90\textheight, keepaspectratio]{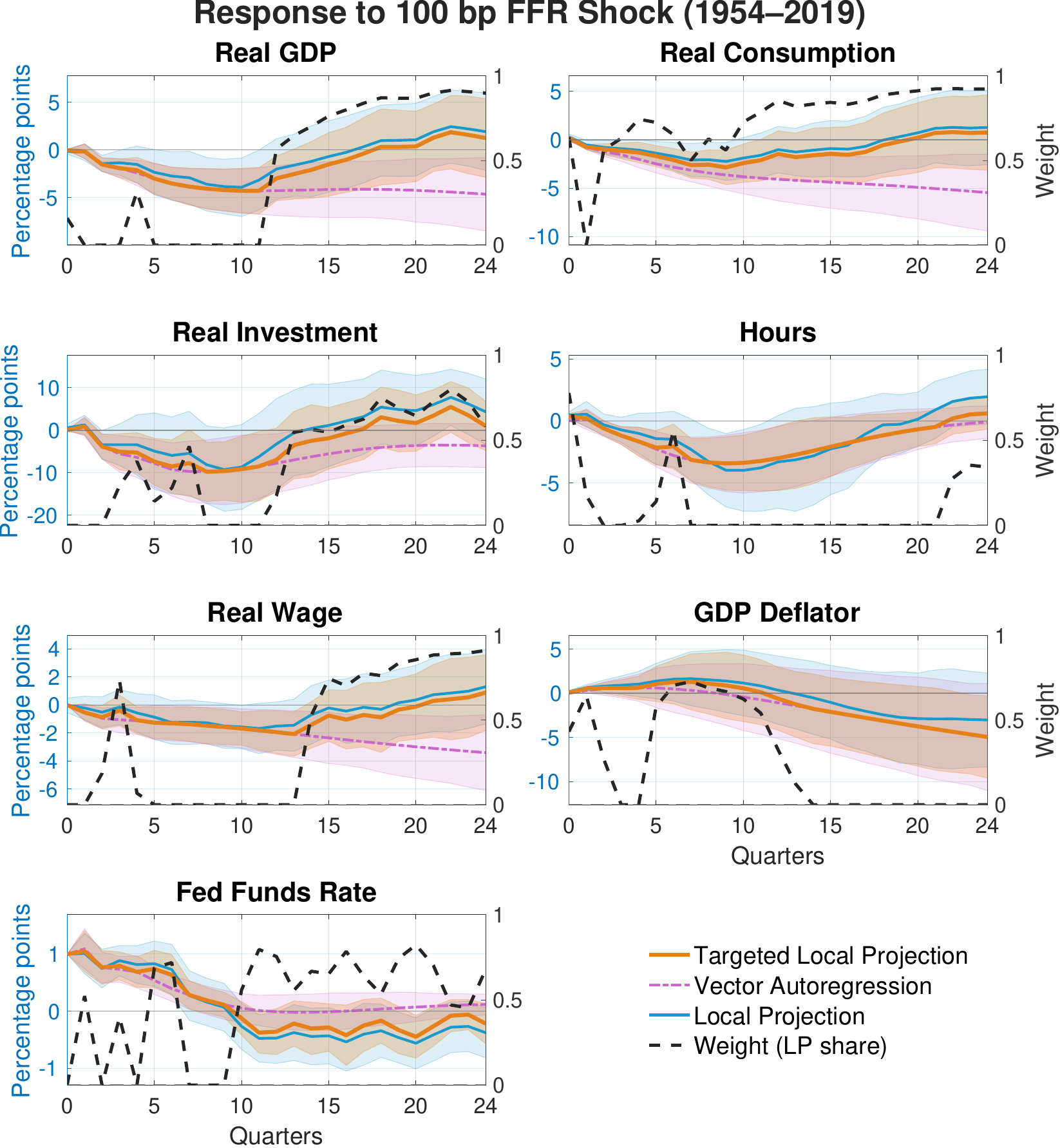}
\caption{Impulse responses to a 100 basis point shock to the federal funds rate (1954--2019). The panels display the responses of Real GDP, Real Consumption, Real Investment, Hours, Real Wage, GDP Deflator, and the federal funds rate. Shaded regions correspond to 90\% confidence intervals. Estimation uses $p=10$ and $q=4$ lags.}
\label{fig:fedfundsrate_irfs}
\end{figure}

\begin{table}[H]
\centering
\resizebox{\textwidth}{!}{
\begin{tabular}{l *{12}{c}}
\hline\hline
Horizon & \multicolumn{3}{c}{Real GDP} & \multicolumn{3}{c}{Real Consumption} & \multicolumn{3}{c}{Real Investment} & \multicolumn{3}{c}{Hours Worked} \\
$h$ & LP & TLP & VAR & LP & TLP & VAR & LP & TLP & VAR & LP & TLP & VAR \\
\hline
0 & $-0.25^{*}$ & $-0.16$ & $0.00$ & $0.31^{*}$ & $0.22$ & $0.00$ & $0.63$ & $0.12$ & $0.00$ & $0.61^{*}$ & $0.54^{*}$ & $0.00$ \\
1 & $-0.46$ & $-0.27$ & $-0.15$ & $-0.47^{*}$ & $-0.65^{*}$ & $-0.74^{*}$ & $1.21$ & $0.88$ & $0.88$ & $0.72$ & $0.51$ & $0.14$ \\
2 & $-1.77^{*}$ & $-1.54^{*}$ & $-1.54^{*}$ & $-0.78^{*}$ & $-0.93^{*}$ & $-1.23^{*}$ & $-3.48$ & $\textcolor{blue}{\boldsymbol{-3.85^{*}}}$ & $-3.85^{*}$ & $-0.03$ & $-0.32$ & $-0.61$ \\
3 & $-1.67^{*}$ & $-1.92^{*}$ & $-1.92^{*}$ & $-0.89^{*}$ & $-1.06^{*}$ & $-1.54^{*}$ & $-3.27$ & $-4.47$ & $-5.58^{*}$ & $-0.34$ & $-0.61$ & $-1.18^{*}$ \\
4 & $-1.79$ & $\textcolor{blue}{\boldsymbol{-2.41^{*}}}$ & $-2.41^{*}$ & $-1.10^{*}$ & $-1.30^{*}$ & $-2.02^{*}$ & $-2.88$ & $-4.13$ & $-6.40^{*}$ & $-0.56$ & $-0.88$ & $-1.68^{*}$ \\
5 & $-2.47^{*}$ & $-3.02^{*}$ & $-3.02^{*}$ & $-1.31^{*}$ & $-1.54^{*}$ & $-2.49^{*}$ & $-4.08$ & $-5.65$ & $-7.82^{*}$ & $-0.99$ & $-1.38$ & $-2.27^{*}$ \\
6 & $-2.63^{*}$ & $-3.48^{*}$ & $-3.52^{*}$ & $-1.56^{*}$ & $-1.84^{*}$ & $-2.83^{*}$ & $-4.98$ & $-6.62$ & $-9.26^{*}$ & $-1.26$ & $-1.67$ & $-2.79^{*}$ \\
7 & $-2.69$ & $\textcolor{blue}{\boldsymbol{-3.47^{*}}}$ & $-3.85^{*}$ & $-1.82^{*}$ & $-2.16^{*}$ & $-3.15^{*}$ & $-4.93$ & $-6.61$ & $-9.73^{*}$ & $-2.08$ & $\textcolor{blue}{\boldsymbol{-2.76^{*}}}$ & $-3.14^{*}$ \\
8 & $-3.11^{*}$ & $-4.08^{*}$ & $-4.08^{*}$ & $-1.73^{*}$ & $-2.07^{*}$ & $-3.43^{*}$ & $-6.81$ & $\textcolor{blue}{\boldsymbol{-9.79^{*}}}$ & $-9.79^{*}$ & $-3.01^{*}$ & $-3.35^{*}$ & $-3.35^{*}$ \\
9 & $-3.14^{*}$ & $-4.22^{*}$ & $-4.22^{*}$ & $-1.84^{*}$ & $-2.23^{*}$ & $-3.64^{*}$ & $-6.84$ & $\textcolor{blue}{\boldsymbol{-9.63^{*}}}$ & $-9.63^{*}$ & $-3.42^{*}$ & $-3.42^{*}$ & $-3.42^{*}$ \\
10 & $-3.15^{*}$ & $-4.29^{*}$ & $-4.29^{*}$ & $-1.64$ & $\textcolor{blue}{\boldsymbol{-2.02^{*}}}$ & $-3.82^{*}$ & $-5.72$ & $\textcolor{blue}{\boldsymbol{-9.10^{*}}}$ & $-9.18^{*}$ & $-3.27^{*}$ & $-3.38^{*}$ & $-3.38^{*}$ \\
11 & $-2.65$ & $\textcolor{blue}{\boldsymbol{-3.65^{*}}}$ & $-4.30^{*}$ & $-1.38$ & $-1.76$ & $-3.96^{*}$ & $-4.06$ & $-6.95$ & $-8.53^{*}$ & $-3.25$ & $\textcolor{blue}{\boldsymbol{-3.23^{*}}}$ & $-3.23^{*}$ \\
12 & $-1.94$ & $\textcolor{blue}{\boldsymbol{-2.74^{*}}}$ & $-4.28^{*}$ & $-0.97$ & $-1.33$ & $-4.08^{*}$ & $-2.27$ & $-4.85$ & $-7.79^{*}$ & $-3.03$ & $\textcolor{blue}{\boldsymbol{-3.00^{*}}}$ & $-3.00^{*}$ \\
13 & $-1.71$ & $\textcolor{blue}{\boldsymbol{-2.53^{*}}}$ & $-4.24^{*}$ & $-1.31$ & $-1.75$ & $-4.19^{*}$ & $-0.01$ & $-2.23$ & $-7.01^{*}$ & $-3.07$ & $\textcolor{blue}{\boldsymbol{-2.72^{*}}}$ & $-2.72^{*}$ \\
14 & $-1.41$ & $-2.22$ & $-4.20^{*}$ & $-1.08$ & $-1.52$ & $-4.28^{*}$ & $0.56$ & $-1.86$ & $-6.26$ & $-3.06$ & $-2.41$ & $-2.41$ \\
15 & $-0.98$ & $-1.74$ & $-4.16^{*}$ & $-0.93$ & $-1.39$ & $-4.38^{*}$ & $1.13$ & $-1.46$ & $-5.56$ & $-2.63$ & $-2.08$ & $-2.08$ \\
16 & $-0.42$ & $-1.11$ & $-4.14^{*}$ & $-0.77$ & $-1.23$ & $-4.47^{*}$ & $2.42$ & $-0.00$ & $-4.96$ & $-2.31$ & $-1.76$ & $-1.76$ \\
\hline\hline
\end{tabular}}
\caption{Impulse Responses to 100 bp Federal Funds Rate Shock (bold blue indicates TLP(8,4) significance where LP(8) is insignificant; a star indicates 90\% significance). The table is truncated to include the first 16 quarters.}
\label{tab:irf_sig}
\end{table}

\begin{table}[H]
\centering
\resizebox{\textwidth}{!}{
\begin{tabular}{l *{12}{c}}
\hline\hline
Horizon & \multicolumn{3}{c}{Real GDP} & \multicolumn{3}{c}{Real Consumption} & \multicolumn{3}{c}{Real Investment} & \multicolumn{3}{c}{Hours Worked} \\
$h$ & LP & TLP & VAR & LP & TLP & VAR & LP & TLP & VAR & LP & TLP & VAR \\
\hline
0 & $-0.17$ & $-0.03$ & $0.00$ & $0.29$ & $0.19$ & $0.00$ & $0.55$ & $0.00$ & $0.00$ & $0.47^{*}$ & $0.36^{*}$ & $0.00$ \\
1 & $-0.24$ & $-0.15$ & $-0.15$ & $-0.53^{*}$ & $-0.74^{*}$ & $-0.74^{*}$ & $1.24$ & $0.88$ & $0.88$ & $0.55$ & $0.21$ & $0.14$ \\
2 & $-1.44^{*}$ & $-1.54^{*}$ & $-1.54^{*}$ & $-0.78^{*}$ & $-0.97^{*}$ & $-1.23^{*}$ & $-3.44^{*}$ & $-3.85^{*}$ & $-3.85^{*}$ & $-0.31$ & $\textcolor{blue}{\boldsymbol{-0.61^{*}}}$ & $-0.61$ \\
3 & $-1.35^{*}$ & $-1.92^{*}$ & $-1.92^{*}$ & $-0.92^{*}$ & $-1.15^{*}$ & $-1.54^{*}$ & $-3.43$ & $\textcolor{blue}{\boldsymbol{-5.08^{*}}}$ & $-5.58^{*}$ & $-0.72$ & $\textcolor{blue}{\boldsymbol{-1.18^{*}}}$ & $-1.18^{*}$ \\
4 & $-1.53$ & $\textcolor{blue}{\boldsymbol{-2.14^{*}}}$ & $-2.41^{*}$ & $-1.07^{*}$ & $-1.31^{*}$ & $-2.02^{*}$ & $-3.44$ & $\textcolor{blue}{\boldsymbol{-5.28^{*}}}$ & $-6.40^{*}$ & $-1.01$ & $\textcolor{blue}{\boldsymbol{-1.66^{*}}}$ & $-1.68^{*}$ \\
5 & $-2.35$ & $\textcolor{blue}{\boldsymbol{-3.02^{*}}}$ & $-3.02^{*}$ & $-1.38^{*}$ & $-1.69^{*}$ & $-2.49^{*}$ & $-4.93$ & $\textcolor{blue}{\boldsymbol{-7.41^{*}}}$ & $-7.82^{*}$ & $-1.43$ & $\textcolor{blue}{\boldsymbol{-2.15^{*}}}$ & $-2.27^{*}$ \\
6 & $-2.75^{*}$ & $-3.52^{*}$ & $-3.52^{*}$ & $-1.70^{*}$ & $-2.09^{*}$ & $-2.83^{*}$ & $-5.95$ & $\textcolor{blue}{\boldsymbol{-8.54^{*}}}$ & $-9.26^{*}$ & $-1.48$ & $\textcolor{blue}{\boldsymbol{-2.06^{*}}}$ & $-2.79^{*}$ \\
7 & $-2.92^{*}$ & $-3.85^{*}$ & $-3.85^{*}$ & $-2.09^{*}$ & $-2.62^{*}$ & $-3.15^{*}$ & $-5.46$ & $-7.75$ & $-9.73^{*}$ & $-2.30$ & $\textcolor{blue}{\boldsymbol{-3.14^{*}}}$ & $-3.14^{*}$ \\
8 & $-3.61^{*}$ & $-4.08^{*}$ & $-4.08^{*}$ & $-2.06^{*}$ & $-2.57^{*}$ & $-3.43^{*}$ & $-8.14$ & $\textcolor{blue}{\boldsymbol{-9.79^{*}}}$ & $-9.79^{*}$ & $-3.27^{*}$ & $-3.35^{*}$ & $-3.35^{*}$ \\
9 & $-3.85^{*}$ & $-4.22^{*}$ & $-4.22^{*}$ & $-2.24^{*}$ & $-2.85^{*}$ & $-3.64^{*}$ & $-9.28$ & $\textcolor{blue}{\boldsymbol{-9.63^{*}}}$ & $-9.63^{*}$ & $-3.97^{*}$ & $-3.42^{*}$ & $-3.42^{*}$ \\
10 & $-3.92^{*}$ & $-4.29^{*}$ & $-4.29^{*}$ & $-1.89$ & $\textcolor{blue}{\boldsymbol{-2.41^{*}}}$ & $-3.82^{*}$ & $-8.67$ & $\textcolor{blue}{\boldsymbol{-9.18^{*}}}$ & $-9.18^{*}$ & $-3.98^{*}$ & $-3.38^{*}$ & $-3.38^{*}$ \\
11 & $-3.17^{*}$ & $-4.30^{*}$ & $-4.30^{*}$ & $-1.60$ & $-2.11$ & $-3.96^{*}$ & $-6.25$ & $\textcolor{blue}{\boldsymbol{-8.53^{*}}}$ & $-8.53^{*}$ & $-3.78^{*}$ & $-3.23^{*}$ & $-3.23^{*}$ \\
12 & $-2.01$ & $\textcolor{blue}{\boldsymbol{-2.99^{*}}}$ & $-4.28^{*}$ & $-1.02$ & $-1.47$ & $-4.08^{*}$ & $-3.26$ & $-6.99$ & $-7.79^{*}$ & $-3.30$ & $\textcolor{blue}{\boldsymbol{-3.00^{*}}}$ & $-3.00^{*}$ \\
13 & $-1.61$ & $-2.54$ & $-4.24^{*}$ & $-1.27$ & $-1.79$ & $-4.19^{*}$ & $-0.67$ & $-3.54$ & $-7.01^{*}$ & $-3.13$ & $\textcolor{blue}{\boldsymbol{-2.72^{*}}}$ & $-2.72^{*}$ \\
14 & $-1.20$ & $-2.10$ & $-4.20^{*}$ & $-1.08$ & $-1.61$ & $-4.28^{*}$ & $0.41$ & $-2.50$ & $-6.26$ & $-2.83$ & $-2.41$ & $-2.41$ \\
15 & $-0.69$ & $-1.52$ & $-4.16^{*}$ & $-0.91$ & $-1.45$ & $-4.38^{*}$ & $1.12$ & $-1.91$ & $-5.56$ & $-2.26$ & $-2.08$ & $-2.08$ \\
16 & $-0.26$ & $-1.04$ & $-4.14^{*}$ & $-0.97$ & $-1.56$ & $-4.47^{*}$ & $2.17$ & $-0.76$ & $-4.96$ & $-1.95$ & $-1.76$ & $-1.76$ \\
\hline\hline
\end{tabular}}
\caption{Impulse Responses to 100 bp Federal Funds Rate Shock (bold blue indicates TLP(10,4) significance where LP(10) is insignificant; a star indicates 90\% significance). The table is truncated to include the first 16 quarters.}
\label{tab:irf_sig104}
\end{table}

\section{Conclusion}
\label{sec:conclusion}

This paper introduces targeted local projections, a frequentist shrinkage estimator 
that averages across LP and VAR impulse responses with data-driven, horizon-specific weights. 
By shrinking LP toward VAR at each horizon, TLP reduces variance at longer horizons 
where LP estimates become noisy, while the optimal weights ensure that the estimator 
places more weight on the asymptotically unbiased LP when the discrepancy between estimators is 
large. The shrinkage parameter has a closed-form solution, making implementation 
straightforward. For inference, we propose the Mean Symmetric Double Bootstrap, which improves 
coverage relative to the standard single bootstrap procedures. 

\vspace{0.2in}

\noindent \textbf{Acknowledgements. }We are grateful for useful comments from Luca Fanelli, Mikkel Plagborg-M{\o}ller, Elena Pesavento, and from the participants at the Structural Econometrics Group seminar in Tilburg 2026 and at the Netherlands Econometrics Study Group and the IAAE conferences in 2026.

\vspace{0.2in}

\noindent \textbf{Data availability and AI usage.} Code and data are available at \url{https://github.com/anemtyrev/Targeted_Local_Projections}. AI tools used: Claude, Claude Code, ChatGPT, Codex. AI usage: coding, proofreading. 

\small

 \nocite{racine1997feasible}
\bibliographystyle{chicago}
\bibliography{bibl}

\newpage 

\appendix
 \normalsize
 \section{The Importance of Projecting Out Controls in Regularized Local Projections}
\label{app:controls}

All estimators discussed in this paper—$\hat{\beta}^{LP}_h$, $\hat{\beta}^{TLP}_h$, 
$\hat{\beta}^{SLP}_h$, and $\hat{\beta}^{VAR}_h$—are scalars. This is achieved by 
projecting out control variables $W$ prior to estimation. Without this step, 
regularization becomes problematic for two reasons.

First, if we apply uniform regularization to all parameters including controls, substantial further derivations are required to identify the target for the control coefficients. 

Second, if we regularize only the parameter of interest while setting zero targets 
for controls, the objective function for TLP becomes
\begin{align}
    Q_n(\theta) = \|\check{Y}_h - \theta_h \check{X}\|^2 
    + \tilde{\lambda} \|\theta_h - \hat{\theta}_h^{VAR}\|^2_\Lambda
\end{align}
where $\check{X} = [Y_2 \; W]$, $\check{Y}_h = Y_{1,h}$, and 
\begin{align}
    \hat{\theta}_h^{VAR} = \begin{bmatrix}
        \hat{\beta}_h^{VAR} \\
        0 \\ 
        \vdots \\
        0
    \end{bmatrix}_{(pk+1)\times 1}, \quad
    \Lambda = \begin{bmatrix}
        1 & 0 & \cdots & 0 \\
        0 & 0 & \cdots & 0 \\
        \vdots & \vdots & \ddots & \vdots \\
        0 & 0 & \cdots & 0
    \end{bmatrix}_{(pk+1)\times (pk+1)}.
\end{align}
The solution is
\begin{align}
    \hat{\theta}^{\text{TLP}}_h = {\phi}(\tilde{\lambda}) \hat{\theta}_h^{LP} 
    + (I_{pk+1} - {\phi}(\tilde{\lambda})) \hat{\theta}_h^{VAR},
\end{align}
where ${\phi}(\tilde{\lambda}) = (\check{X}' \check{X} + \tilde{\lambda} \Lambda)^{-1}(\check{X}' \check{X})$. 
The matrix ${\phi}(\tilde{\lambda})$ is not diagonal, so any $\tilde{\lambda} \neq 0$ introduces 
bias into the control coefficients. 

\cite{yamada2017frisch} shows that projecting out controls preserves both the 
parameter of interest and its variance, up to a scaling factor, in ridge regression (but not in generalized ridge regression). 
This approach eliminates the targeting ambiguity and prevents cross-contamination 
between the parameter of interest and controls. For this reason, projecting out controls is 
essential for both TLP and SLP estimators.

\section{Proofs} \label{app:proofs}

\noindent \textbf{Lemma 1.} Under Assumptions 1 and 2, $y_t$, $u_t= \Gamma(\epsilon_t + T^{-\zeta} \alpha(L) \epsilon_t)$, and $u_t-\Gamma \epsilon_t$ are $\mathcal L_d$ near-epoch dependent (NED) of size $m^{-1/2}$ on the $\alpha$-mixing process $\epsilon_t$ of size $-r/(r-2)$, where recall $d=4+2\delta$, and $r=2+\delta=d/2$. \\

\noindent \textbf{Proof of Lemma 1.} Under Assumption 1, $y_t$ can be inverted into a VMA($\infty$). Letting $\alpha(L)=\sum_{\ell=1}^{\infty} \alpha_\ell L^{\ell}$, with the convention $\alpha_0=0$,

\begin{equation*}
    y_t = \sum_{j=0}^{\infty} (A^j \Gamma + T^{-\zeta} A^j \Gamma \alpha_j) \epsilon_{t-j} \equiv \sum_{j=0}^{\infty} B_{T,j} \epsilon_{t-j}.
\end{equation*}
By Assumption 1, 
$
    \sup_t E \|\epsilon_t\|_d < M,
$
for a generic constant $M$ - which we use throughout to denote various uniform bounds - 
and the VMA coefficients are absolutely summable 
$
    \sup_T \sum_{j=0}^{\infty}\|B_{T,j}\| < \infty,
$
where $\| \cdot \|$ denotes the Euclidean norm for vectors, and the sup-norm for matrices.
Since 
$
     E(y_t|\mathcal F_{t-m}^{t+m}) = E(y_t|\mathcal F_{t-m})
    =
    \sum_{j=0}^{m} B_{T,j}\epsilon_{t-j},
$
\begin{align*}
    \|y_t-E(y_t|\mathcal F_{t-m}^{t+m})\|_d
    &=
    \left(
    \mathbb E
    \left\|
    \sum_{j=m+1}^{\infty} B_{T,j}\epsilon_{t-j}
    \right\|^d
    \right)^{1/d} \le
     \left(\sum_{j=m+1}^{\infty}
    \|B_{T,j}\|^d
    \mathbb E \|\epsilon_{t-j}\|^d
    \right)^{1/d} \\
    &\le
    \left(M
    \sum_{j=m+1}^{\infty}\|B_{T,j}\|^d\right)^{1/d} = M^{1/d} \left(
    \sum_{j=m+1}^{\infty}\|B_{T,j}\|^d\right)^{1/d}. 
\end{align*}
Since $||B_{T,j}||$ are absolutely summable, they are also $d>4$ summable. Moreover, $
    \sum_{j=m+1}^{\infty}\|B_{T,j}\|^d \leq \| \Gamma\|^d\left(\sum_{j=m+1}^{\infty}\|A^{jd}\| \times \|1+ \alpha_j T^{-\zeta} \|^d \right) \leq M \left(\sum_{j=m+1}^{\infty}\|A^j\|^d \right)$. Since $\sum_{j=m+1}^{\infty}\|A^j\|^d  \rightarrow 0$ as $m\rightarrow \infty$, $y_t$ is $\mathcal L_d$ NED on $\epsilon_t$,
with approximation constants $c_t=M$ and  $\nu_{m,1} = \left(
    \sum_{j=m+1}^{\infty}\|B_{T,j}\|^d\right)^{1/d} $. Since $A^m$ decays faster than any polynomial rate, $y_t$ is $\mathcal L_d$ NED of every finite polynomial size, therefore also of size $m^{-1}$ and $m^{-1/2}$. The result for $u_t$ can be obtained following the same steps as for $y_t$ but without needing the VMA($\infty$) inversion. Noting that $
     E(u_t|\mathcal F_{t-m}^{t+m}) = E(u_t|\mathcal F_{t-m})
    =
    \sum_{j=0}^{m} \Gamma (1+T^{-\zeta}) \alpha_j\epsilon_{t-j}$,
we have:
\begin{align*}
   \|u_t-E(u_t|\mathcal F_{t-m}^{t+m})\|_d
    &=
    \left(
    \mathbb E
    \left\|
    \sum_{j=m+1}^{\infty} \Gamma (1+T^{-\zeta}) \alpha_j\epsilon_{t-j}
    \right\|^d
    \right)^{1/d} \\
    &\le
     \left(\sum_{j=m+1}^{\infty}
    \|\Gamma (1+T^{-\zeta}) \alpha_j\|^d
    \mathbb E \|\epsilon_{t-j}\|^d
    \right)^{1/d} \le M
    \left(
    \sum_{j=m+1}^{\infty}\|\alpha_j\|^d\right)^{1/d} \\
    &\equiv M \nu_{m,2}.
\end{align*}
Therefore, $u_t$ is $\mathcal L_d$ NED on the $\alpha$-mixing process $\epsilon_t$, with approximation constants 
 $c_t= M$ and $\nu_{m,2}$ which also decay exponentially by Assumption 1; hence, the NED order of $u_t$ is also $m^{-1}$ or just $m^{-1/2}$. Moreover, $u_t-\Gamma \epsilon_t$ inherits the dependence properties of $u_t$ by similar derivations as above. $\square$

\noindent \textbf{Lemma 2 (WLLN).} If $\{h_t\}$ is a mean zero random vector and $\mathcal L_2$ NED on an $\alpha$-mixing process $v_t$ of size $\nu_m=m^{-1/2}$ with constants $c_t=M$ fixed for a generic $M$, with $\sup_t \mathbb E|v_{t,i}|^{2+\delta} < M<\infty$, where $v_{t,i}$ is the $i^{th}$ element of $v_t$, then $T^{-1} \sum_{t=1}^T h_t \stackrel{p}{\rightarrow} 0.$ 

\textbf{Proof of Lemma 2.} From Proposition 2.9 in \cite{wooldridge1988some}, $h_t$ is an $\mathcal L_2$ mixingale of size $-1/2$ with constants $c_t=M$  and $\psi_m=5 \alpha_{[m/2]}^{1/2-1/d} +\nu{[m/2]}$. Hence, it is also an $\mathcal L_1$ mixingale of the same size with the same constants, and $\lim \sum_{T=1}^T c_t =M<\infty$. Hence, by \cite{andrews1988laws}, Theorem 1, the desired result follows. $\square$

\noindent \textbf{Lemma 3 (CLT).} If $\{h_t\}$ is a mean zero random vector and $\mathcal L_2$ NED on an $\alpha$-mixing process $v_t$ of size $-r^*/(r^*+2)$, with $\sup_t \mathbb E|v_{t,i}|^{r^*} < M<\infty$ for some $r^*>2$, where $v_{t,i}$ is the $i^{th}$ element of $v_t$, and the long-run variance exists and it is positive definite: $\lim_{T\rightarrow \infty} \mbox{Var} \left(T^{-1/2} \sum_{t=1}^T h_t\right) = V$, then $T^{-1/2} \sum_{t=1}^T h_t \stackrel{d}{\rightarrow} \mathcal N(0, V)$.

\textbf{Proof of Lemma 3}. Since $h_t$ is an $\mathcal L_1$ mixingale of size $-1/2$ with constants defined as in the proof of Lemma 2, the result follows by applying Theorem 2.11 in \cite{wooldridge1988some}, as already verified in \cite{hall2012inference}.

\noindent \textbf{Proof of Theorem 1. } 
The asymptotic bias for the VAR estimator was derived in \cite{olea2024double} under Assumption 1. For the LP estimator, consider the regression coefficient $\hat \theta_h$ on all coefficients, obtained with regression model:
\begin{align*}
    y_{i,t+h} = y_{j,t} \beta_h + w_{t}'\gamma_h +  u_{t+h}^h \equiv z_t ' \theta_h +  u_{t+h}^h,
\end{align*}
where $w_t$ contains all the $p$ lags of $y_t$ included in LP, $z_t = (y_{j,t}, w_t')'$, $\theta_h = (\beta_h, \gamma_h')'$, $\gamma_h$ are functions of the original VAR coefficients obtained by backward substituting the model in \eqref{DGP}, and $u_{t+h}^h$ are the errors which can be shown, by backward substitution, to be linear combinations of $u_{t+h}, \ldots, u_t$, the reduced-form errors of the VAR. Denote this linear combination by 
\begin{align}
 u_{t+h}^h  = \sum_{n=0}^h \Phi_n' u_{t+h-n}   = \sum_{n=0}^h \Phi_n' \Gamma \epsilon_{t+h-n} + T^{-\zeta} \sum_{n=0}^h \Phi_n' \Gamma \alpha(L) \epsilon_{t+h-n}.
\end{align}
Note that
\begin{align}\label{eq:LPdecomp}
    \hat \theta_{h}^{LP} &= \theta_h + \left( T^{-1} \sum_{t=1}^T z_t z_t'\right)^{-1} \left( T^{-1} \sum_{t=1}^T z_t  u_{t+h}^h\right) \\ \nonumber
    & = \theta_h + \left( T^{-1} \sum_{t=1}^T z_t z_t'\right)^{-1} \left( T^{-1/2} \sum_{n=0}^h z_t \Phi_n' \Gamma \epsilon_{t+h-n} + T^{-1/2-\zeta} \sum_{n=0}^h z_t \Phi_n' \Gamma \alpha(L) \epsilon_{t+h-n}\right).
\end{align}
Since $\epsilon_t$ is MDS by Assumption 1, and $z_t$ either contains lags of $y_t$ or $y_{j,t}$, while $\Gamma$ is a lower-triangular matrix, $\mathbb E(z_t \Phi_n' \Gamma\epsilon_{t+h}) = 0$ for all $h\geq 0$, hence 
$$
\mathbb E\left( T^{-1/2} \sum_{n=0}^h z_t \Phi_n' \Gamma \epsilon_{t+h-n}\right) =0. 
$$
Let $z_{i,t}$ be the $i^{th}$ element of $z_t$, and denote by $\Psi_{n,\ell,k^*}$ the $(k^*)^{th}$ element of $\Psi_{n,\ell}\equiv \Phi_n' \Gamma \alpha_\ell$, where $\alpha(L)=\sum_{\ell=1}^{\infty} \alpha_\ell L^\ell$, and $L$ is the lag operator. Then:
\begin{align*} 
T^{-1/2-\zeta} \sum_{n=0}^h \mathbb E (z_t \Phi_n' \Gamma \alpha(L) \epsilon_{t+h-n})  &\leq\mathbb E \left |T^{-1/2-\zeta} \sum_{n=0}^h z_{i,t} \sum_{\ell=1}^{\infty} \Psi_{n,\ell} \epsilon_{t+h-n-\ell} \right | \\
&\leq T^{-1/2-\zeta} \sum_{n=0}^h \sum_{\ell=1}^{\infty} \mathbb E |z_{i,t} \Psi_{n,\ell} \epsilon_{t+h-n-\ell} | \\
&\leq T^{-1/2-\zeta} \sum_{n=0}^h \sum_{\ell=1}^{\infty} \mathbb E \left |\sum_{k^*=1}^k z_{i,t}\Psi_{n,\ell,k^*} \epsilon_{k^*,t+h-n-\ell} \right| \\
& \leq T^{-1/2-\zeta}  \sum_{n=0}^h \sum_{\ell=1}^{\infty} \sum_{k^*=1}^k \times |\Psi_{n,\ell,k^*}| \sup_{t,i} \|z_{i,t}\|_2 \sup_{t,k^*}\|\epsilon_{k^*,t}\|_2.
\end{align*}
By Assumption 1 (i), 
$\sup_{t,k^*} |\epsilon_{k^*,t}|_2 < M$. Also, 
\begin{align}
 \sup_t \|y_t\|_2 &\leq \sum_{j=0}^{\infty} \| B_{T,j} \| \times \sup_t \|\epsilon_t\|_2 = \sum_{j=0}^{\infty} \| B_{T,j} \| \times \left(\mathbb E \left(\sum_{k',k^*=1}^ k \epsilon_{k',t} \epsilon_{k^*,t} \right)^2\right)^{1/2} \\
 & = \sum_{j=0}^{\infty} \| B_{T,j} \| \times \left( \sum_{k^*=1}^ k \mathbb E(\epsilon_{k^*,t}^2) \right)^{1/2} = \sum_{j=0}^{\infty} \| B_{T,j} \| \times
 \left(\sum_{k^*=1}^k \sigma_{k^*}^2\right)^{1/2} <M,
\end{align}
where the first inequality is the triangle inequality, the second and third equalities follow from Assumption 1(i), and the last inequality from Assumption 1 (i), (ii) and (v). Hence,
$\sup_t \|y_t\|_2 < M$, and 
\begin{align*} \mathbb E \left |T^{-1/2-\zeta} \sum_{n=0}^h z_{i,t} \sum_{\ell=1}^{\infty} \Psi_{n,\ell} \epsilon_{t+h-n-\ell} \right | & \leq T^{-1/2-\zeta} k M \sum_{\ell=1}^{\infty} \max_{k^*} |\Psi_{n,\ell,k^*}|  = o(T^{-1/2}),
\end{align*}
where the equality follows from $|\Psi_{n,\ell,k^*}|$ being absolutely summable. Hence,
\begin{equation}\label{eq:crossmomzu}
A \equiv \mathbb E \left(T^{-1} \sum_{t=1}^T z_t  u_{t+h}^h\right) = o(T^{-1/2}).
\end{equation}

We now verify Lemma 2, WLLN, for $T^{-1} \sum_{t=1}^T z_t z_t'$. Because $y_t$ is $\mathcal L_d$ NED of size $-1/2$ on $\epsilon_t$ by Lemma 1, so is $z_t$, and  so $z_t z_t'$ but with lower norm $\mathcal L_{r}$ NED, by Theorems 17.8-17.10 in \cite{davidson1994stochastic}, $r=2+\delta$, hence also with $\mathcal L_2$ norm. By H\"older's inequality, $\sup_t E| z_{i,t} z_{i',t}|^{2+\delta}< M$, with $z_{i,t}$ the $i^{th}$ element of $z_t$. We now show that $\mathbb E(y_t y_{t-\ell}') = E(\tilde y_t \tilde y_{t-\ell}) + o(1)$. Since 
$y_t = \tilde y_t + T^{-\zeta}\bar u_t$, where $\bar u_t = \Gamma \alpha(L) \epsilon_t$, and recall that $\tilde y_t$ is the process in \eqref{DGP} purged of misspecification, hence setting $\zeta=0$. Then we have:
\begin{align*}
    \mathbb E(y_t y_{t-\ell}) &=\mathbb E [(\tilde y_t + T^{-\zeta}\bar u_t)(\tilde y_{t-\ell} + T^{-\zeta}\bar u_{t-\ell})]  \\&= \mathbb E( \tilde y_t \tilde y_{t-\ell}) + T^{-\zeta}\mathbb E (y_t \bar u_{t-\ell}) + T^{-\zeta}\mathbb E (y_{t-\ell} \bar u_t)  + T^{-2\zeta}\mathbb E (y_{t-\ell} \bar u_{t-\ell}).
\end{align*}
By similar arguments as for $z_t \Phi_n' \Gamma \epsilon_{t+h-n}$,
$\sup_{t,\ell}\mathbb E\|y_{t-\ell} \bar u_t\| < M $, 
$\sup_{t,\ell}\mathbb E \| y_t \bar u_{t-\ell}\| <M$, and $\sup_{t,\ell}\mathbb E \|y_{t-\ell} \bar u_{t-\ell}\| <M$. Hence, $ \mathbb E(y_t y_{t-\ell}) = \mathbb E(\tilde y_t \tilde y_{t-\ell}) +o(1)$, and because $z_t$ contains $y_{j,t}$ and lags of $y_t$,
\begin{equation*}
    \mathbb E(z_t z_t') = \mathbb E(\tilde z_t \tilde z_t') +o(1) \equiv \tilde S_q + o(1), 
\end{equation*}
where $\tilde z_t$ is the process with lags of $\tilde y_t$ replacing lags of $y_t$, and $\tilde S_q$ originates from calculating the autocovariances of $\tilde y_t$, a stationary VAR(1) process by Assumption 1. 
Hence, by the WLLN in Lemma 2,
\begin{equation}\label{eq:secondmomz}
    B \equiv T^{-1} \sum_{t=1}^T z_t z_t' \stackrel{p}{\rightarrow} \tilde S_q.
\end{equation}

Now $E(B \times A) = E(B) \times E(A) + Cov (B, A)$, and $|Cov(B,A)| \leq \sqrt{Var(B) \times Var(A)}$. By \eqref{eq:secondmomz}, $Var(B) = o(1)$, and by the CLT shown in the proof of Theorem 2 below, $T^{1/2} Var(A) = O(1)$. Hence, using \eqref{eq:crossmomzu} into \eqref{eq:LPdecomp}, $E(B \times A) = \tilde S_q \times o(T^{-1/2}) + o(1) \times O(T^{-1/2}) = o(T^{-1/2})$, and we obtain the desired result:
\begin{equation*}
\mathbb E (\hat \beta_h^{LP}) = \beta_h + o(T^{-1/2}). \qquad \square
\end{equation*}

\noindent \textbf{Proof of Theorem 2. } 
We first derive the asymptotic distribution for the LP estimator. Since we already have the LP estimator decomposition in \eqref{eq:LPdecomp}, and we know $T^{-1} \sum_{t=1}^T z_t z_t' \stackrel{p}{\rightarrow} \tilde S_q$ by \eqref{eq:secondmomz}, applying the CLT in Lemma 3 will deliver the asymptotic distribution of $T^{1/2}(\hat \beta_h^{LP} - \beta_h).$ Therefore, we now verify the conditions in Lemma 3 for 
\begin{equation*}
 z_t u_{t+h}^h =  \sum_{n=0}^h \Phi_n' u_{t+h-n} =\sum_{n=0}^h z_t \Phi_n' \Gamma (\epsilon_{t+h-n} + T^{-1/2}\alpha(L) \epsilon_{t+h-n}).
\end{equation*} 

By Lemma 1, $u_t$ is $\mathcal L_d$ NED of size $m^{-1/2}$ on $\epsilon_t$, hence so is $u_{t+h}^h$ by Theorem 17.8 in \cite{davidson1994stochastic}. Also, since $y_t$ shares this NED property, so does $z_t$, and, by Theorem 17.10 in \cite{davidson1994stochastic}, so does $z_t u_{t+h}^h$ except it is $\mathcal L_2$ NED (actually, $\mathcal L_{r}$, but   $r=2+\delta>2$, so $\mathcal L_2$ suffices).  By arguments similar to bounding $\mathbb E\|y_t\|^2$ in the proof of Theorem 1, one can show that $\sup_t \mathbb E |z_{i,t}  u_{t+h}^h|^{r}< M$, where $r=d/2=2+\delta$. We now derive the long-run variance of $z_t u_{t+h}^h$. To that end, consider first the long-run variance of $y_{j,t} u_{t+h}^h$. 
\begin{align*}
    \lim_{T\rightarrow \infty} \mbox{Var} (T^{-1/2}\sum_{t=1}^T y_{j,t}  u_{t+h}^h)  & = T^{-1} \sum_{t,s=1}^T \mathbb E(y_{j,t}  u_{t+h}^h  u_{s+h}^h y_{j,s} ).
\end{align*}
Recall that $y_{j,t} = \tilde y_{j,t} + T^{-1/2} \bar u_t$, and let $u_t = \tilde u_t + \bar v_t$, where $\tilde u_t = \Gamma \epsilon_t$, and $\bar v_t = \Gamma \sum_{l=1}^{\infty} \alpha_\ell \epsilon_{t-\ell}$. By the same arguments as for the proof of Lemma 1, $\tilde u_t$, $\tilde y_t$, $\bar u_t$, $\bar v_t$ are $\mathcal L_d$ NED of the same size on an $\alpha$-mixing process (in fact, $\tilde u_t$ is also of smaller order because it is a martingale difference by Assumption 1). Therefore, so are their cross-products, but by Theorem 17.10 in \cite{davidson1994stochastic}, the relevant NED norm for the cross-products is $\mathcal L_{1+\delta/2}$. So, the terms in $\sum_{t,s=1}^T\mathbb E(y_{j,t}  u_{t+h}^h u_{s+h}^h y_{j,s} ) = \sum_{n,n^*=0}^{h} \Phi_n'\mathbb E(y_{j,t} u_{t+h-n} u_{s+h-n^*}^h y_{j,s})\Phi_{n^*}$, pre-multiplied by $T^{-1/2}$, are necessarily of a smaller order than \\ $\sum_{t,s=1}^T\sum_{n,n^*=0}^{h} \Phi_n'\mathbb E(\tilde y_{j,t} \tilde u_{t+h-n} \tilde u_{s+h-n^*}^h y_{j,s})\Phi_{n^*}$. 

It remains to show that $\sum_{t,s=1}^T\sum_{n,n^*=0}^{h}  \Phi_n' \mathbb E(\tilde y_{j,t} \tilde u_{t+h-n} \tilde u_{s+h-n^*}^h \tilde y_{j,s}) \Phi_{n^*}$ exists. 
By Lemma 1 and Theorem 17.9, the product $y_{j,t} \tilde u_{t+h-n}$ is also an $\mathcal L_2$ NED sequence on an $\alpha$-mixing process of size $m^{-r/(r-2)}$, with moments $r+\delta> r$ existing by Assumption 1 and H\"older's inequality. Therefore, setting $r=2+\delta$ in Theorem 17.7 in \cite{davidson1994stochastic}, we obtain that $\sum_{t,s=1}^T\sum_{n,n^*=0}^{h} \Phi_n'\mathbb E(\tilde y_{j,t} \tilde u_{t+h-n} \tilde u_{s+h-n^*}^h \tilde y_{j,s}) \Phi_{n^*}$. Hence,
\begin{align*}
    \lim_{T\rightarrow \infty} \mbox{Var} (T^{-1/2}\sum_{t=1}^T y_{j,t} u_{t+h}^h) = \sum_{n,n^*=0}^{h}  \Phi_n' \left[T^{-1} \sum_{t,s=1}^T E(\tilde y_{j,t} \tilde u_{t+h-n} \tilde u_{s+h-n^*}^h \tilde y_{j,s}) \right] \Phi_{n^*} \equiv V_{j,h}< \infty,
\end{align*}
Therefore, by Lemma 3 (CLT), $T^{-1/2}\sum_{t=1}^T y_{j,t} u_{t+h}^h \stackrel{d}{\rightarrow} \mathcal N(0,V_{j,h})$. 

By similar arguments, $T^{-1/2}\sum_{t=1}^T z_t u_{t+h}^h \stackrel{d}{\rightarrow} \mathcal N(0,V_h^{LP})$, where
\begin{equation*} 
V_h^{LP} = \lim_{T\rightarrow \infty} Var \left( T^{-1/2} \sum_{t=1}^T \tilde z_t \tilde u_{t+h}^h\right),
\end{equation*}
where $\tilde u_{t+h}^h = \sum_{n=0}^h \Phi_n' \tilde u_{t+h-n}$. Hence, 
\begin{equation}\label{eq:varianceLP}
    T^{1/2} (\hat \beta_h^{LP}-\beta_h) \stackrel{d}{\rightarrow}\mathcal N(0, \tilde S_q^{-1} V_h^{LP} \tilde S_q^{-1}),
\end{equation}
where $\Sigma_h^{LP} \equiv \tilde S_q^{-1} V_h^{LP} \tilde S_q^{-1}$.  Equation \eqref{eq:varianceLP} shows that the asymptotic variance of the LP estimator is unaffected by the misspecification. The VAR estimator has asymptotic bias, as shown in Theorem 1; however, for the limiting variance calculations, the asymptotic bias is irrelevant. Therefore, by similar arguments as for LP, 
\begin{equation}\label{eq:varianceVAR}
T^{1/2}(\hat \beta_h^{VAR}-\beta_h) \stackrel{d}{\rightarrow} \mathcal N(aBias_h, \Sigma_h^{VAR}),
\end{equation}
and $\Sigma_h^{VAR}$ is the asymptotic variance of standard structural VARs obtained from the data generating process $y_t = Ay_{t-1} + \Gamma \epsilon_t$, thus also unaffected by misspecification. $\square$.

\noindent \textbf{Proof of Theorem 3. } 

\textit{Risk decomposition.}
For TLP at horizon $h$, the risk function is
\begin{align}
    R_h(\lambda_h)
    &= T\,\E\!\left[\big(\hat{\beta}^{TLP}_h(\lambda_h) - \beta_h\big)^2\right] \nonumber \\
    &= \underbrace{T \big(\E[\hat{\beta}^{TLP}_h(\lambda_h)] - \beta_h\big)^2}_{\text{bias}^2(\lambda_h)}
    + \underbrace{T\,\mbox{Var}[\hat{\beta}^{TLP}_h(\lambda_h)-\beta_h]}_{\text{variance}(\lambda_h)}. \label{eq:risk}
\end{align}
Let $b_h = \plim [v(\lambda_h)]$, which exists for each $\lambda_h$ since  $T^{-1} X'X \stackrel{p}{\rightarrow} E(x_t x_t')$ by Assumptions \ref{asn:1}-\ref{asn:2} and the LLN for $\alpha$-mixing processes. Moreover, by the LLN, $\sup_{\lambda_h \in [0, \infty)}|v(\lambda_h)- b_h| \stackrel{p}{\rightarrow} 0$.
From Theorem \ref{thm:1}, the squared bias is:
\begin{align}\nonumber
   T \big(\E[\hat{\beta}^{TLP}_h(\lambda_h)] - \beta_h\big)^2 &= T \big(\E[(v(\lambda_h) \hat \beta_h^{LP}  + (1-v(\lambda_h))\hat \beta_h^{VAR} - \beta_h]\big)^2 \\ \nonumber
   & = T \big( \E[ v(\lambda_h)(\hat \beta_h^{LP}- \beta_h)] +   \E[(1- v(\lambda_h))(\hat \beta_h^{VAR}- \beta_h)] \big)^2 \\\nonumber
   & = b_h^2\, (\E (T^{1/2}[ \hat \beta_h^{LP}- \beta_h]))^2 + (1-b_h)^2  \, (\E (T^{1/2}[ \hat \beta_h^{VAR}- \beta_h]))^2 \\ \nonumber
   & + 2 b_h \,(1-b_h) \,\E[ T^{1/2}(\hat \beta_h^{LP}- \beta_h)] \, \E[ T^{1/2}(\hat \beta_h^{VAR}- \beta_h)] +o(1)\\ \nonumber
   & = b_h^2 \, \times 0 + (1-b_h)^2\,  C_h^2 +  2 b_h \,(1-b_h) \times  0 \times C_h +o(1) \\ 
   & =  (1-b_h)^2   C_h^2 +o(1).\label{eq:bias_tot} 
\end{align}
where $C_h = \textup{aBias}_h$.
From Theorem \ref{thm:2}, the variance term in the risk criterion is:
\begin{align} \nonumber
T \mbox{Var} [\hat{\beta}^{TLP}_h(\lambda_h)]&= T \mbox{Var} [b_h\hat{\beta}^{LP}_h(\lambda_h) + (1-b_h) \hat{\beta}^{VAR}_h(\lambda_h)] \\ \nonumber
& = b_h^2 \, \mbox{Var} (T^{1/2}[\hat{\beta}^{LP}_h(\lambda_h) -\beta_h ]) + (1-b_h)^2 \, \mbox{Var} ( T^{1/2}[\hat{\beta}^{VAR}_h(\lambda_h) -\beta_h]) \\ \nonumber
& + 2 \, b_h\, (1-b_h) \, \mbox{Cov} (T^{1/2}[\hat{\beta}^{LP}_h(\lambda_h) -\beta_h ], T^{1/2}[\hat{\beta}^{VAR}_h(\lambda_h) -\beta_h]\, ) +o(1) \\ 
& = b_h^2 \,\Sigma_h^{LP} +  (1-b_h)^2 \,\Sigma_h^{VAR} +  2 \, b_h\, (1-b_h) \, \Sigma_h^{COV}. \label{eq:var_tot}
\end{align}
Substituting \eqref{eq:bias_tot}-\eqref{eq:var_tot} into \eqref{eq:risk}, we obtain:
\begin{align} \label{eq:risk2}
    R_h(\lambda_h) &=  (1-b_h)^2\, C_h^2 +\,[b_h^2\,\Sigma_h^{LP} +(1-b_h)^2\,\Sigma_h^{VAR} + \, 2 b_h\, (1-b_h) \, \Sigma_h^{COV} ] +o(1).
\end{align}

\textit{Empirical risk criterion.}
From Theorem \ref{thm:2},
\[
\hat{C}_h = T^{1/2}\big(\hat{\beta}_h^{LP} - \hat{\beta}_h^{VAR}\big) \stackrel{d}{\rightarrow} \mathcal N( -C_h, \Sigma_h^{LP} + \Sigma_h^{VAR} - 2 \Sigma_h^{COV}).
\]

It follows that 
\begin{align}
   \E[\hat{C}_h^2] = C_h^2 + \Sigma_h^{LP} + \Sigma_h^{VAR} - 2 \Sigma_h^{COV} +o(1).
\end{align}
Hence, $C_h^2 = \E[\hat{C}_h^2]- \Sigma_h^{LP} - \Sigma_h^{VAR} +2 \Sigma_h^{COV} +o(1)$.
Since $\hat \Sigma_h^{\mu} \stackrel{p}{\rightarrow} \Sigma_h^{\mu}$ for $\mu \in \{LP, VAR, COV\}$, 
combining the estimated bias and variance terms, the asymptotically unbiased empirical risk criterion is:
\begin{align} \nonumber
    \hat{R}_h(\lambda_h)
    &= T(1 - v(\lambda_h))^2
        (\hat{\beta}_h^{LP} - \hat{\beta}_h^{VAR})^2 - (1 - v(\lambda_h))^2 \hat \Sigma_h^{LP} - (1 - v(\lambda_h))^2 \hat \Sigma_h^{VAR} \\ \nonumber
        &+2 (1 - v(\lambda_h))^2 \hat \Sigma_h^{COV}+
        v(\lambda_h)^2 \hat \Sigma_h^{LP}
        + (1 - v(\lambda_h))^2 \hat \Sigma_h^{VAR} \\ \nonumber
       & +2 v(\lambda_h)(1 - v(\lambda_h)) \hat \Sigma_h^{COV} \\
    & = T(1 - v(\lambda_h))^2
        (\hat{\beta}_h^{LP} - \hat{\beta}_h^{VAR})^2 + (2 v(\lambda_h)-1) \hat \Sigma_h^{LP} + 2(1 - v(\lambda_h)) \hat \Sigma^{COV}_h.
\end{align}

Since $\hat R_h(\lambda_h)$ is continuous in $\lambda_h$, $\sup_{\lambda_h \in (0, \infty)} |\E[\hat R_h(\lambda_h)] - R_h(\lambda_h)| \rightarrow 0$. $\square$ \\

\noindent \textbf{Proof of Theorem 4}. 
Since there is a one-to-one mapping between $\lambda_h$ and $v(\lambda_h) = (X'X + \lambda_h)^{-1} X'X$, optimizing $\hat R_h(\lambda_h)$ with respect to $\lambda_h$ is equivalent to optimizing it with respect to $v(\lambda_h)$. Let $\hat A_h = T(\hat{\beta}_h^{LP} - \hat{\beta}_h^{VAR})^2$. Then 
\begin{align*}
  \frac{\partial \hat{R}_h(\lambda_h) }{\partial v(\lambda_h)} & = -2 (1 - v(\lambda_h)) \hat A_h +2 (\hat \Sigma_h^{LP} -\hat \Sigma_h^{COV}) =0\\
  \Leftrightarrow v(\lambda_h) \hat A_h &= \hat A_h - (\hat \Sigma_h^{LP} -\hat \Sigma_h^{COV}) \\
   \Leftrightarrow v(\lambda_h) &= 1- \frac{(\hat \Sigma_h^{LP} -\hat \Sigma_h^{COV})}{\hat A_h}.
\end{align*}
Since $\hat \Sigma_h^{LP} - \hat \Sigma_h^{COV} >0$, the weights are in finite samples less than one. However, they also need to be positive according to the definition $v(\lambda_h) = (X'X + \lambda_h)^{-1} X'X$. Therefore, the optimal weights are:
\begin{align*}
v(\hat \lambda_h) &= \mbox{max}\left [1- \frac{\hat \Sigma_h^{LP} -\hat \Sigma_h^{COV}}{\hat A_h},0\right].
\end{align*}
with $\hat\lambda_h \rightarrow \infty$ when $v(\hat \lambda_h)=0$, and 
\begin{align*}
    T^{-1} \hat \lambda_h = \frac{T^{-1}X'X (\hat \Sigma_h^{LP} -\hat \Sigma_h^{COV})}{\hat A_h- (\hat \Sigma_h^{LP} -\hat \Sigma_h^{COV})}.
\end{align*}  
otherwise. $\square$

\section{Additional simulations}\label{app:sims}
\renewcommand{\thefigure}{C\arabic{figure}}
\subsection{Larger misspecification $\eta$ \label{app:large}}
\setcounter{figure}{0}
\newpage
\begin{figure}[H]
    \centering
    \includegraphics[width=\linewidth, height=0.90\textheight, keepaspectratio]{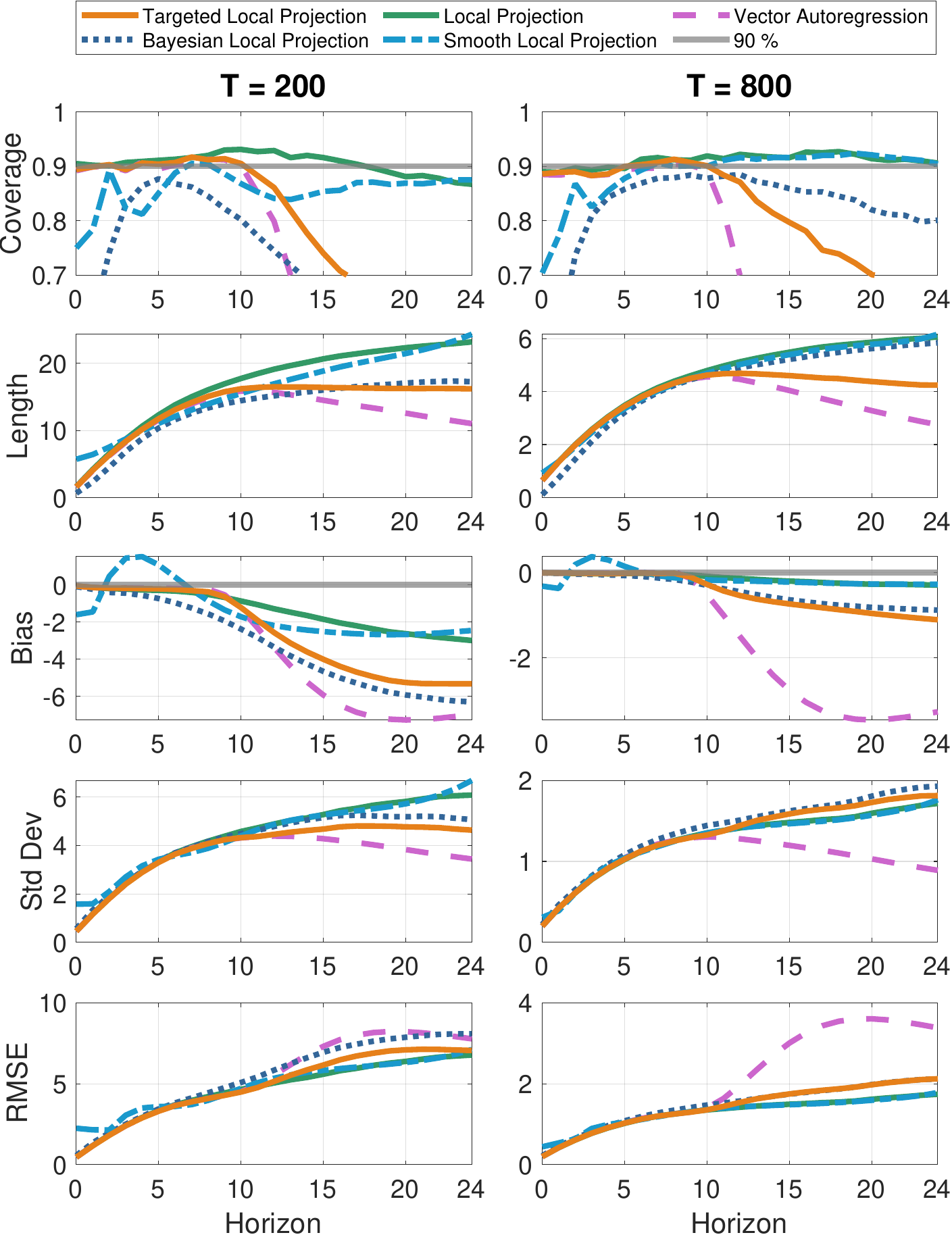}
    \caption{Coverage, length, bias, standard deviation, and RMSE for $T=200$ (left) and $T=800$ (right), DGP = VARMA(1,$100$), $\eta=8$, conditionally homoskedastic errors. Methods: Targeted Local Projections (10,8), Local Projections (10), Vector Autoregression (8), Smooth Local Projections (10) and Bayesian Local Projections (8,8). Inference by MSDB for TLP, LP, VAR and SLP.}
    \label{fig:varma1inftyBB8_metrics}
\end{figure}

\newpage
\begin{figure}[H]
    \centering
    \includegraphics[width=\linewidth, height=0.90\textheight, keepaspectratio]{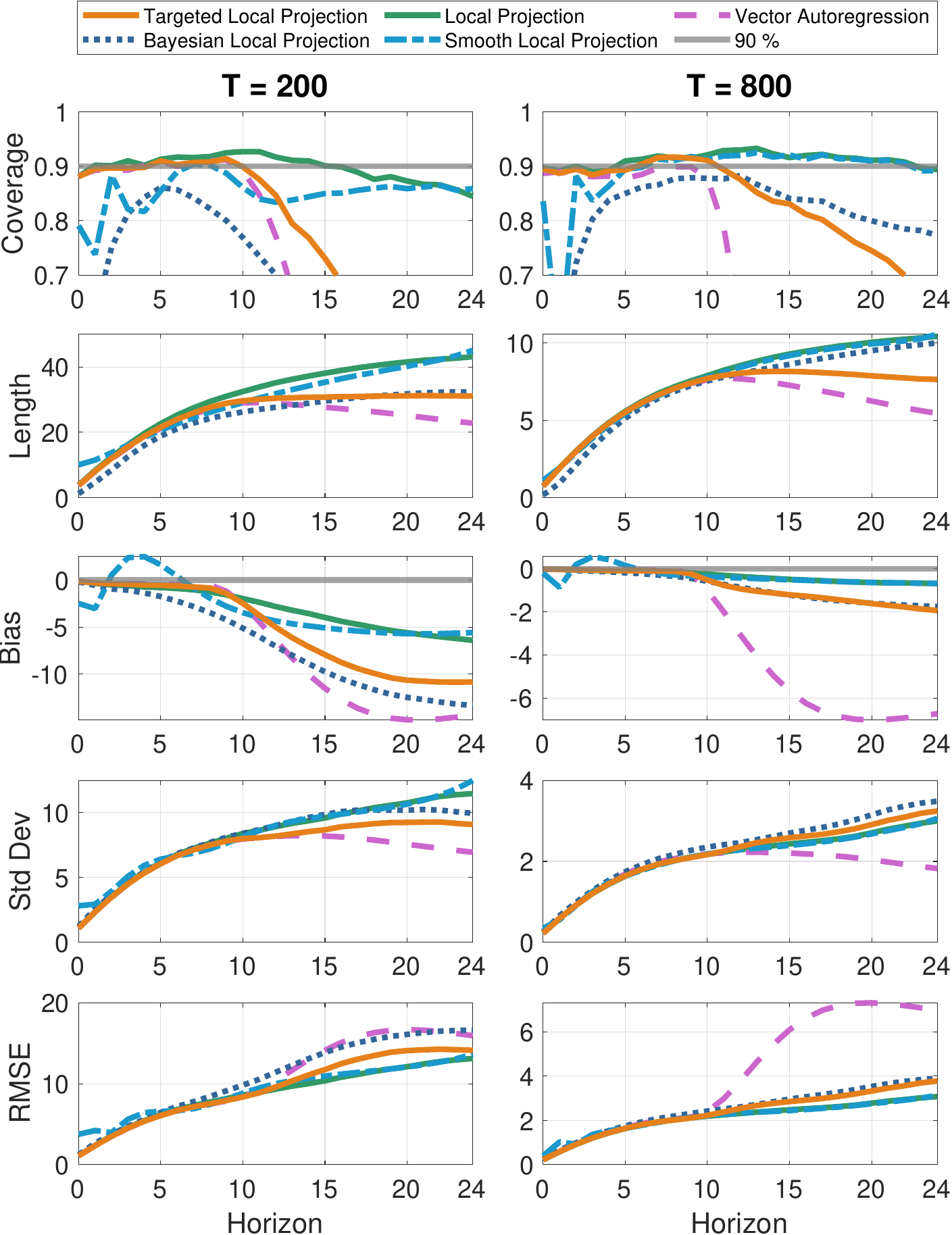}
    \caption{Coverage, length, bias, standard deviation, and RMSE for $T=200$ (left) and $T=800$ (right), DGP = VARMA(1,$100$), $\eta=16$, conditionally homoskedastic errors. Methods: Targeted Local Projections (10,8), Local Projections (10), Vector Autoregression (8), Smooth Local Projections (10) and Bayesian Local Projections (8,8). Inference by MSDB for TLP, LP, VAR and SLP.}
    \label{fig:varma1inftyBB16_metrics}
\end{figure}

\newpage
\begin{figure}[H]
    \centering
    \includegraphics[width=\linewidth, height=0.90\textheight, keepaspectratio]{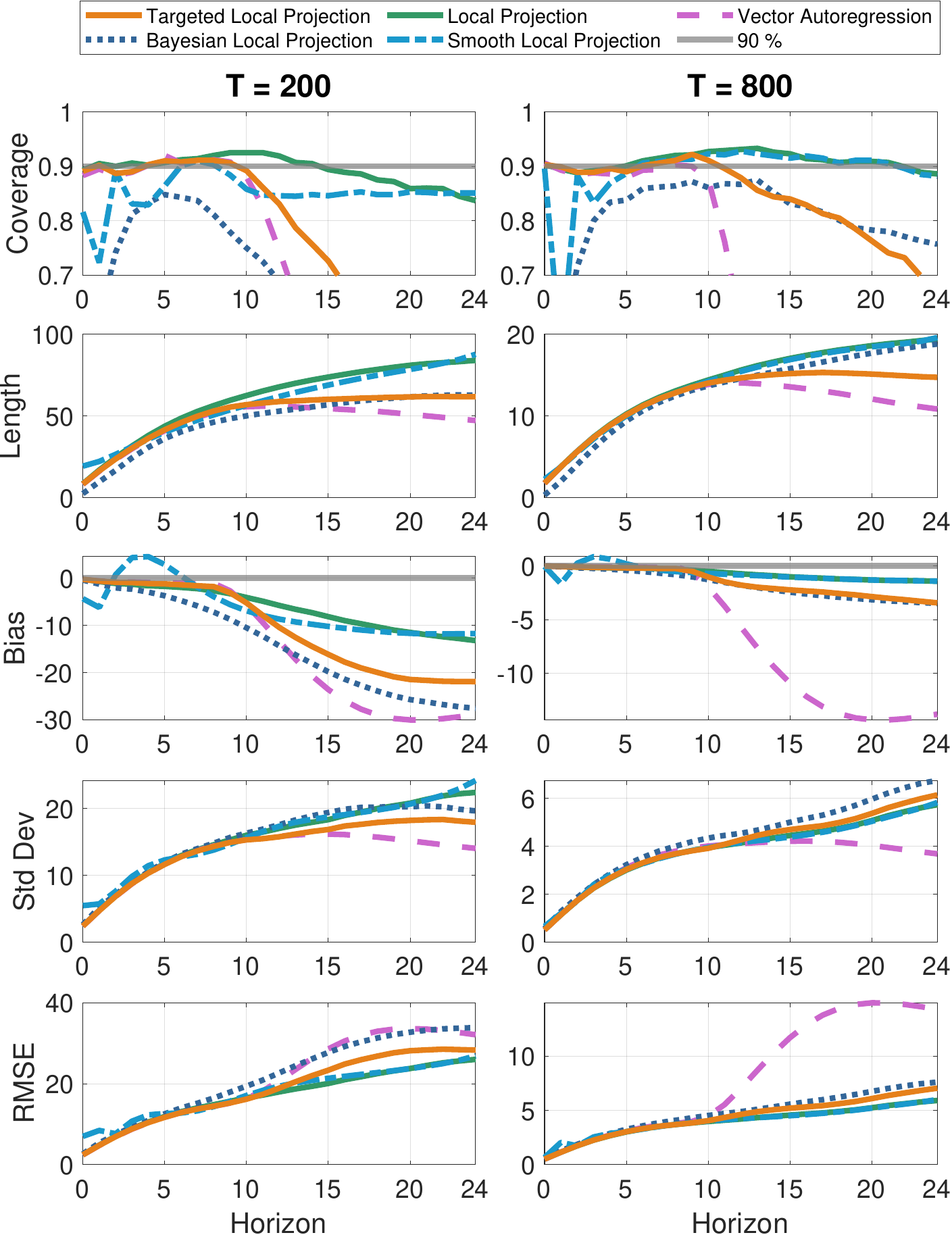}
    \caption{Coverage, length, bias, standard deviation, and RMSE for $T=200$ (left) and $T=800$ (right), DGP = VARMA(1,$100$), $\eta=32$, conditionally homoskedastic errors. Methods: Targeted Local Projections (10,8), Local Projections (10), Vector Autoregression (8), Smooth Local Projections (10) and Bayesian Local Projections (8,8). Inference by MSDB for TLP, LP, VAR and SLP.}
    \label{fig:varma1inftyBB32_metrics}
\end{figure}

\subsection{GARCH(1,1) errors\label{app:garch}}
\begin{figure}[H]
    \centering
    \includegraphics[width=0.98\linewidth, keepaspectratio]{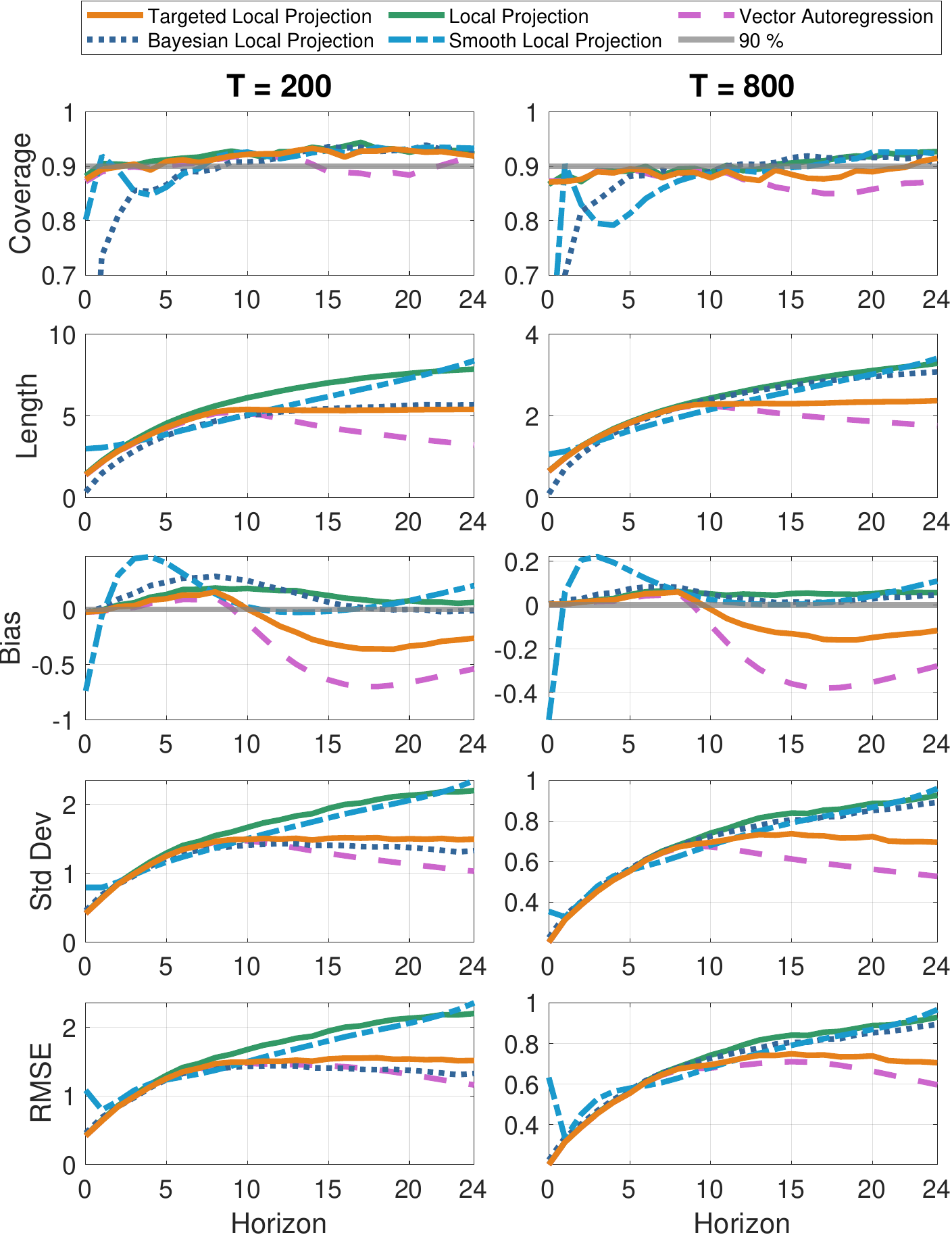}
    \caption{Coverage, length, bias, standard deviation, and RMSE for $T=200$ (left) and $T=800$ (right), DGP = VARMA(1,$100$), $\eta=1$, \textbf{GARCH(1,1) errors}. Methods: Targeted Local Projections (10,8), Local Projections (10), Vector Autoregression (8), Smooth Local Projections (10) and Bayesian Local Projections (8,8). Inference by MSDB for TLP, LP, VAR and SLP.}
    \label{fig:varma1garch_metrics}
\end{figure}
\newpage
\begin{figure}[H]
    \centering
    \includegraphics[width=\linewidth, height=0.90\textheight, keepaspectratio]{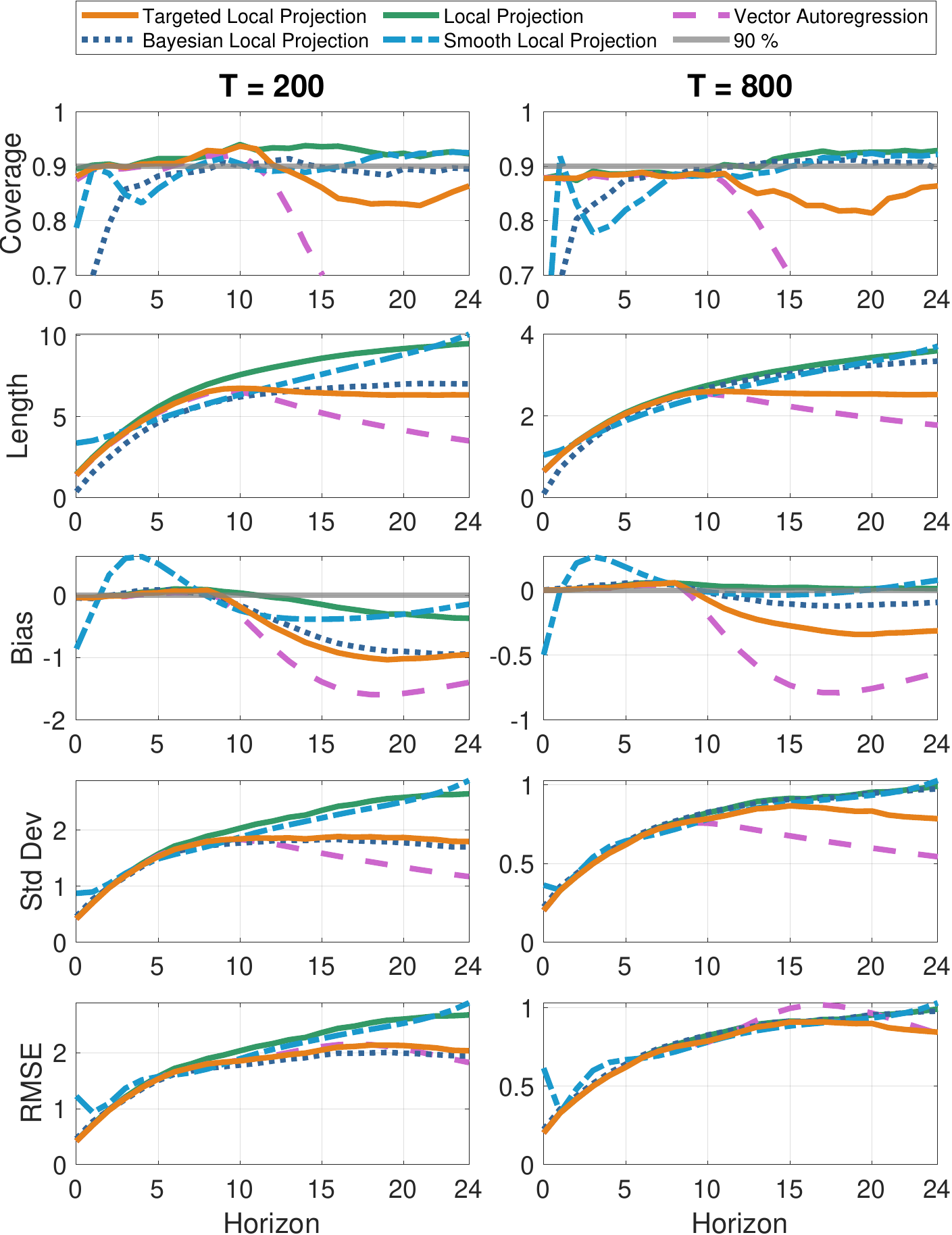}
    \caption{Coverage, length, bias, standard deviation, and RMSE for $T=200$ (left) and $T=800$ (right), DGP = VARMA(1,$100$), $\eta=2$, \textbf{GARCH(1,1) errors}. Methods: Targeted Local Projections (10,8), Local Projections (10), Vector Autoregression (8), Smooth Local Projections (10) and Bayesian Local Projections (8,8). Inference by MSDB for TLP, LP, VAR and SLP.}
    \label{fig:varma1garchBB2_metrics}
\end{figure}
\newpage
\begin{figure}[H]
    \centering
    \includegraphics[width=\linewidth, height=0.90\textheight, keepaspectratio]{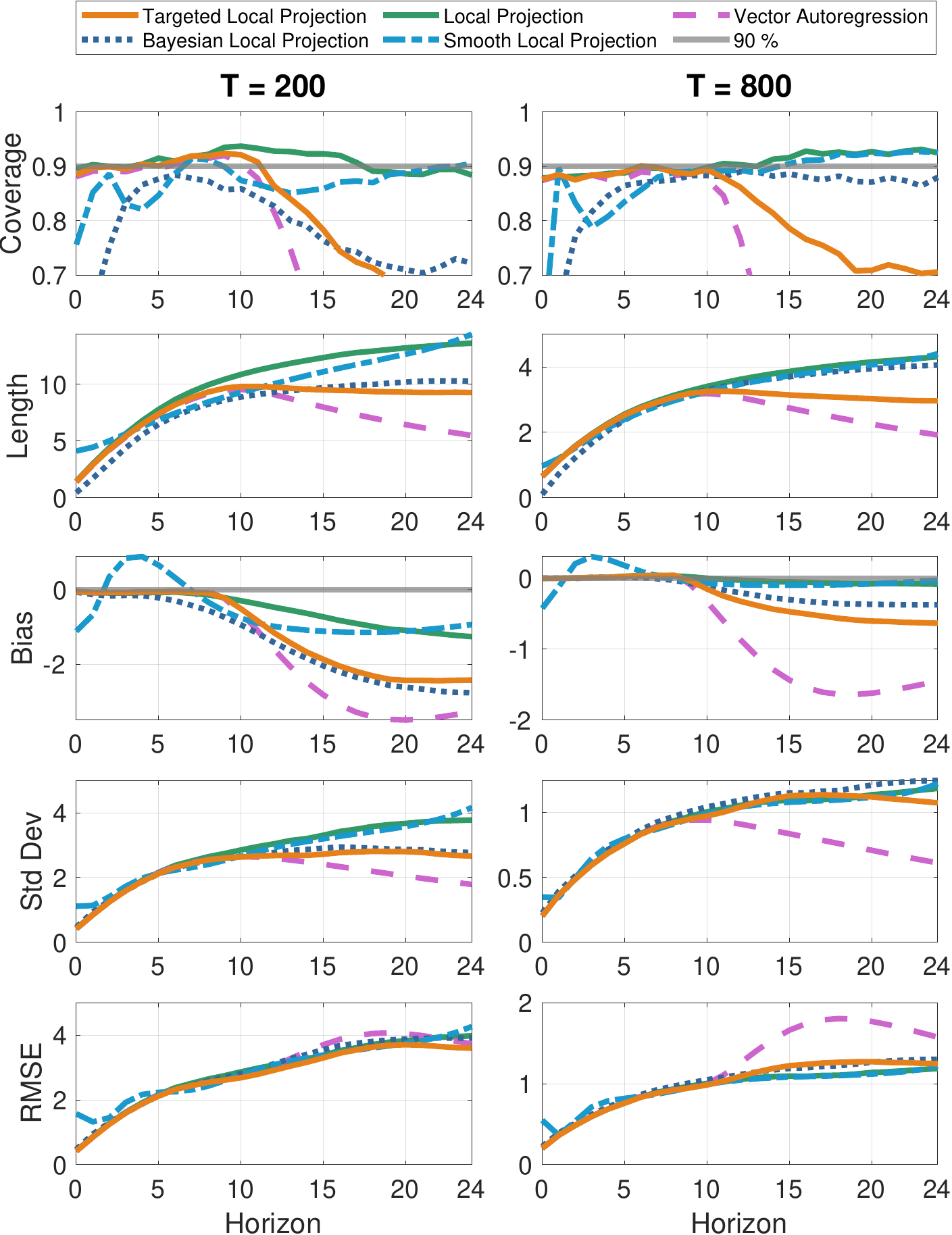}
    \caption{Coverage, length, bias, standard deviation, and RMSE for $T=200$ (left) and $T=800$ (right), DGP = VARMA(1,$100$), $\eta=4$, \textbf{GARCH(1,1) errors}. Methods: Targeted Local Projections (10,8), Local Projections (10), Vector Autoregression (8), Smooth Local Projections (10) and Bayesian Local Projections (8,8). Inference by MSDB for TLP, LP, VAR and SLP.}
    \label{fig:varma1garchBB4_metrics}
\end{figure}
\restoregeometry
\section{Additional empirical results}\label{app:empirics}
\renewcommand{\thefigure}{D\arabic{figure}}
\setcounter{figure}{0}
\begin{figure}[H]
\centering
\includegraphics[width=\linewidth, height=0.90\textheight, keepaspectratio]{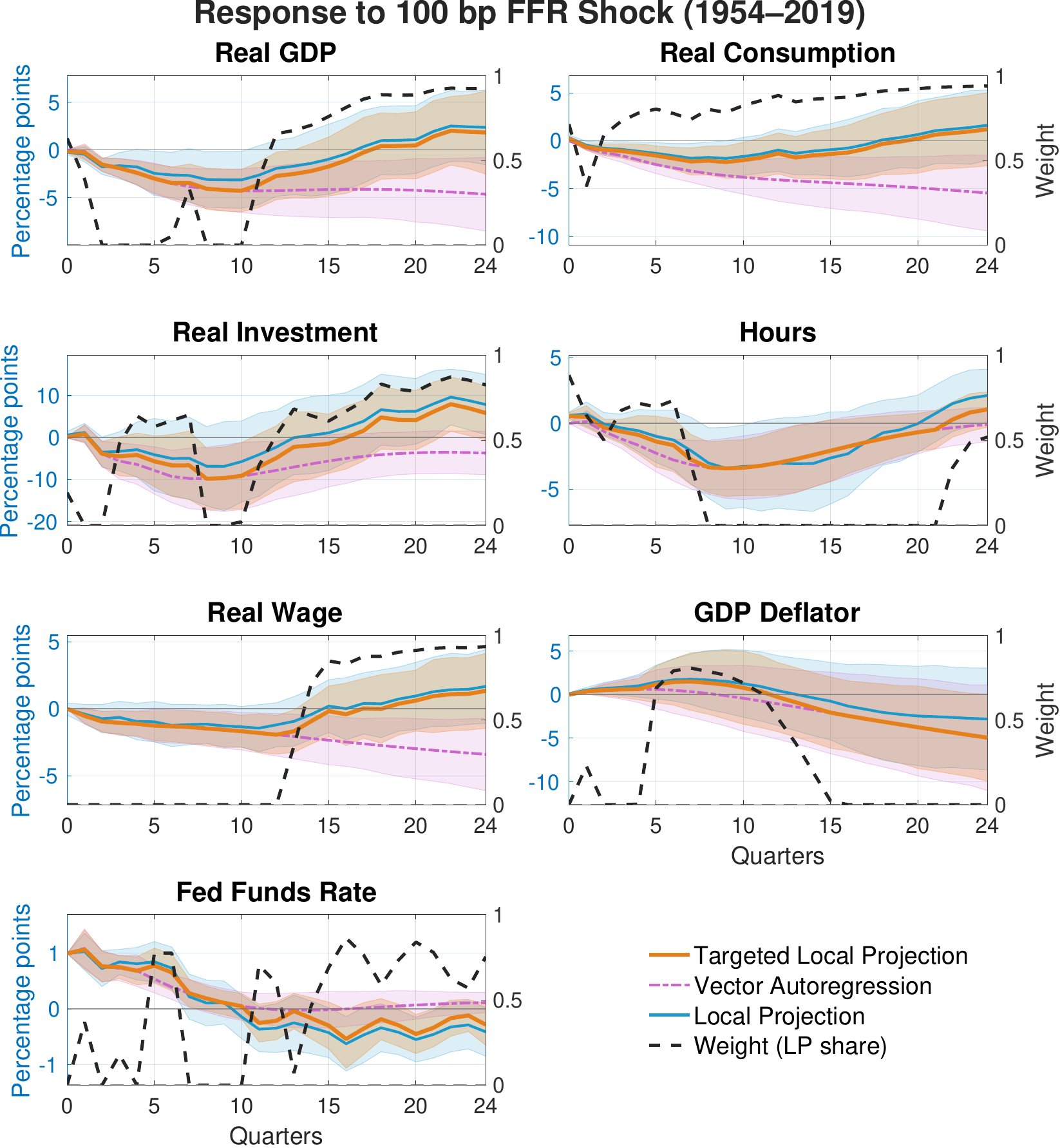}
\caption{Impulse responses to a 100 basis point shock to the federal funds rate (1954--2019). The panels display the responses of Real GDP, Real Consumption, Real Investment, Hours, Real Wage, GDP Deflator, and the federal funds rate. Shaded regions correspond to 90\% confidence intervals. Estimation uses $p=8$ and $q=4$ lags.}
\label{fig:fedfundsrate_irfs2}
\end{figure}
\newpage

\end{document}